\documentclass[11pt,letterpaper]{article}
\usepackage[margin=1in]{geometry}

\usepackage[utf8]{inputenc}
\usepackage[colorlinks=true,allcolors=blue]{hyperref}

\usepackage{amsmath}
\usepackage{amsfonts}
\usepackage{amssymb}

\usepackage{amsthm}

\usepackage{braket}
\usepackage{comment}
\usepackage[colorlinks=true,allcolors=blue]{hyperref}
\usepackage{breakcites}
\usepackage{cleveref}
\usepackage[sanserif,basic]{complexity}
\usepackage{graphicx}
\usepackage[bb=dsserif]{mathalpha}
\usepackage{complexity}
\usepackage{scalerel}
\usepackage[shortlabels]{enumitem}

\usepackage[colorinlistoftodos]{todonotes}

\newtheorem{theorem}{Theorem}[section]

\newtheorem{lemma}[theorem]{Lemma}
\newtheorem{claim}[theorem]{Claim}
\newtheorem{proposition}[theorem]{Proposition}
\newtheorem{corollary}[theorem]{Corollary}
\newtheorem{definition}[theorem]{Definition}

\newtheorem{remark}[theorem]{Remark}

\newtheorem{scheme}[theorem]{Scheme}

\newcommand{\bits}{{\{0,1\}}}
\newcommand{\negl}{{\mathsf{negl}}}
\newcommand{\YES}{{\scaleobj{0.9}{\mathsf{YES}}}}
\newcommand{\NO}{{\scaleobj{0.9}{\mathsf{NO}}}}
\newcommand{\rank}{{\mathsf{rank}}}
\newcommand{\tr}{{\mathsf{tr}}}

\newclass{\clonableQMA}{clonableQMA}

\title{A Computational Separation Between\\ Quantum No-cloning and No-telegraphing}
\author{Barak Nehoran\\Princeton University \and Mark Zhandry\\NTT Research \& Princeton University}
\date{}

\begin{document}
\maketitle

% The prohibitions on both making copies of general quantum states, and sending them over classical channels, two of the fundamental no-go theorems of quantum information,
% are known to be equivalent, in the sense that a collection of quantum states is clonable if and only if it is telegraphable.

\begin{abstract}
Two of the fundamental no-go theorems of quantum information are the no-cloning theorem (that it is impossible to make copies of general quantum states) and the no-teleportation theorem (the prohibition on telegraphing, or sending quantum states over classical channels without pre-shared entanglement). They are known to be equivalent, in the sense that a collection of quantum states is telegraphable if and only if it is clonable.

Our main result suggests that this is not the case when computational efficiency is considered. We give a collection of quantum states and quantum oracles relative to which these states are efficiently clonable but \emph{not} efficiently telegraphable. Given that the opposite scenario is impossible (states that can be telegraphed can always trivially be cloned), this gives the most complete quantum oracle separation possible between these two important no-go properties.

%In doing so, we introduce a related quantum no-go property, \emph{reconstructibility}, which refers to the ability to construct a quantum state from a uniquely identifying classical description. We show the stronger result of a collection of quantum states that are efficiently clonable but not efficiently reconstructible. This novel no-go property only exists in relation to computational efficiency, as it is trivial for unbounded computation. It thus opens up the possibility of further computational no-go properties that have not yet been studied because they do not exist outside the computational context.

We additionally study the complexity class $\clonableQMA$, a subset of $\QMA$ whose witnesses are efficiently clonable. As a consequence of our main result, we give a quantum oracle separation between $\clonableQMA$ and 
the class $\QCMA$, whose witnesses are restricted to classical strings. 
% $\QCMA$, the class of languages with classical witnesses. 
We also propose a candidate oracle-free promise problem separating these classes. We finally demonstrate an application of clonable-but-not-telegraphable states to cryptography, by showing how such states can be used to protect against key exfiltration.
\end{abstract}

% \newpage

\section{Introduction}

One of the defining features of quantum information is the no-cloning theorem: that it is impossible to copy a general quantum state~\cite{Park70,WoottersZurek82,Dieks82}. Another fundamental no-go theorem is the no-teleportation theorem: that it is impossible (without any pre-shared entanglement) to send quantum information over a classical channel~\cite{Werner_1998}. Because of the potential confusion with the very possible task of quantum teleportation~\cite{PhysRevLett.70.1895}, we prefer to use the term \emph{telegraphing} to refer to this latter task.

These two no-go theorems are well-understood to be \emph{equivalent}, in the following sense: given a set of quantum states $S=\{|\psi_1\rangle,|\psi_2\rangle,\cdots\}$, then states in $S$ can be perfectly cloned if and only if they can be perfectly telegraphed, both clonability and telegraphability being equivalent to the states in $S$ being orthogonal 
(see Appendix~\ref{appendix:no-cloning-equals-no-telegraphing} for a formal proof of this fact).
Here, $S$ being cloned means there is a process mapping $|\psi_i\rangle$ to two copies of $|\psi_i\rangle$. $S$ being telegraphed means there is a deconstruction process which maps $|\psi_i\rangle$ into classical information $c_i$, and a reconstruction process that maps $c_i$ back to $|\psi_i\rangle$. 

%In fact, both cloning and telegraphing are possible if and only if the states in $S$ are orthogonal. 
% Relaxed notions of cloning and telegraphing that only require approximate success are also equivalent, both being possible exactly when the states are ``almost'' orthogonal in some sense.

%Telegraphing can be further broken down into two steps: the first is deconstruction, which (potentially probabilistically) maps $|\psi_i\rangle$ to a label $c_i$ which uniquely identifies $|\psi_i\rangle$, in the sense that if $i_1\neq i_2$, the supports of $c_{i_1}$ derived from $|\psi_{i_1}\rangle$ and of $c_{i_2}$ derived from $|\psi_{i_2}\rangle$ are disjoint. The second step is reconstruction, which maps $c_i$ back to $|\psi_i\rangle$. Note that deconstruction alone is also readily seen to be equivalent to cloning and telegraphing, even without explicitly requiring a reconstruction procedure. This is because reconstruction is always trivial: for any set $S$ and unambiguous set of labels $c_i$, there will always be a process that maps $c_i$ back to $|\psi_i\rangle$. 

%Thus, we have that no-cloning, no-telegraphing, and no-deconstruction theorems are all equivalent, while reconstruction is always possible.

\paragraph{Introducing Computational Constraints.} The above discussion is information-theoretic. 
Here, we ask: \emph{what happens when computational constraints are considered}? We consider a set $S$ to be computationally clonable if there is a \emph{polynomial-time} quantum algorithm that solves the cloning task on $S$. Likewise, we consider $S$ to be computationally telegraphable if there is both a polynomial-time deconstruction and corresponding polynomial-time reconstruction procedure for~$S$. 
%$S$ is computationally deconstructible if there is a polynomial-time deconstruction algorithm, with no requirement that reconstruction is efficient. Finally, $S$ is computationally reconstructible if there is an efficient procedure to construct $|\psi_i\rangle$ from \emph{some} unambiguous labelling, even if there is no efficient deconstruction procedure that computes such labels from $|\psi_i\rangle$. Note that, a priori, it may be that $S$ is both computationally deconstructible \emph{and} reconstructible, but \emph{not} computationally telegraphable, if the deconstruction and reconstruction algorithms use incompatible labellings.

%We say that two computational no-go theorems such as no-cloning and no-telegraphing are equivalent if, for any set $S$, $S$ is clonable if and only if it is telegraphable.

We observe the trivial relationship that computational telegraphing implies computational cloning: by running reconstruction twice on the deconstructed classical information $c_i$, one obtains two copies of $|\psi_i\rangle$, therefore cloning. This process is only twice as slow as the original telegraphing procedure, and is therefore efficient if telegraphing is efficient.
However, the converse is a priori unclear: if a state can be cloned efficiently, it is not clear if there is an efficient process to deconstruct the state into a classical $c_i$ and also an efficient process to turn $c_i$ back into the quantum state. 

%But what about the other properties? In general, there are 8 possibilities to consider, specified by whether or not $S$ is computationally clonable, whether or not $S$ is computationally deconstructible, and whether or not $S$ is computationally reconstructible. Are all possibilities possible? Or are some combinations impossible?

\subsection{Our Results}

In this work, we provide evidence that no-cloning and no-telegraphing are \emph{not} equivalent properties in the computationally bounded setting. Our main theorem is:
\begin{theorem}[Informal presentation of Theorem~\ref{thm:clonable_untelegraphable}]\label{thm:maininf} There exists a quantum oracle $\mathcal{O}$ and a set of quantum states $S$ such that $S$ can be efficiently cloned relative to $\mathcal{O}$, but there is no efficient telegraphing procedure relative to $\mathcal{O}$. Even more, there is no telegraphing procedure where the reconstruction is efficient, even if we allow deconstruction to be unbounded.
\end{theorem}
In other words, while no-cloning implies no-telegraphing, the converse is not true, at least relative to a quantum oracle.

Counter-intuitively, we prove this theorem by starting from a certain set of orthogonal but computationally \emph{unclonable} states (related to those used by~\cite{rewinding}). By the trivial relationship that
telegraphing implies clonability,
% unclonability implies no telegraphing, 
we observe that these states cannot be efficiently telegraphed either. 
But of course, while these states can be cloned inefficiently (as they are all orthogonal), we need them to be clonable efficiently. 
We therefore augment the setup with a quantum oracle that performs this cloning in a single query. 
The main technical difficulty is that we need to show that despite adding this cloning oracle, telegraphing remains inefficient. We do this through a multistep process, gradually converting any supposed telegraphing scheme that uses this oracle into a telegraphing scheme that does not, reaching a contradiction.

An interesting consequence of our proof is that the no-telegraphing property holds, \emph{even if the sender is allowed to be inefficient}. The only party that needs to be efficient to achieve a separation is the receiver.

We additionally bring to light certain applications of clonable-but-untelegraphable states to both complexity theory and cryptography.

\paragraph{Complexity Theory.} 
An important open problem in quantum complexity theory is the question of whether quantum states are more powerful than classical strings as proofs (or witnesses) for efficient quantum computation. This is the question of whether the class $\QMA$ of problems which have efficiently verifiable quantum proofs is contained in the class $\QCMA$ of problems where a classical proof suffices~\cite{aharonov2002quantum}. 
A number of recent works~\cite{aaronson2007quantum, fefferman2018quantum, natarajan2022classical, li2023classical}
have endeavored to give increasingly strong oracle separations between the two classes. We take a slightly different approach, inspired by clonable-but-untelegraphable states.
We define a class $\clonableQMA$ of problems which have quantum proofs that are \emph{efficiently clonable}. It is easy to see that $\QCMA \subseteq \clonableQMA \subseteq \QMA$, and we argue that $\clonableQMA$ is not likely equal to either of the other two. Specifically, we use the clonable-but-untelegraphable states of Theorem~\ref{thm:maininf} to show a quantum oracle separation with $\QCMA$:

\begin{theorem}[Informal presentation of Theorem~\ref{thm:clonableqma-qcma-oracle-separation}]\label{thm:complexityinf}
There exists a unitary quantum oracle $\mathcal{O}$ such that $\clonableQMA^\mathcal{O}$ is not contained in $\QCMA^\mathcal{O}$.
\end{theorem}

We also give a candidate \emph{oracle-free} promise problem separating these classes, and we show that any such problem would immediately yield clonable-but-untelegraphable quantum states.
Finally, we argue that it is unlikely that $\QMA$ is contained in $\clonableQMA$, as it would mean that every $\QMA$-complete problem would have efficiently clonable witnesses and act as a barrier against the existence of public-key quantum money.

\paragraph{Cryptography.} 
While no-cloning has seen significant attention in cryptography (e.g.~\cite{Aaronson11,subspaces,ssl,ColLiuLiuZha21}), no-telegraphing has so far received little-to-no attention. We give a proof-of-concept application of clonable-but-untelegraphable states to protecting against key exfiltration. See Section~\ref{par:crypto-applications} below for a more expansive discussion of the result. This motivates the use of no-telegraphing as an useful cryptographic tool.

\subsection{Motivation}

The importance of the interplay between quantum information and computational complexity is becoming increasingly clear. For example, computational complexity played a crucial role in Harlow and Hayden's resolution to the black hole Firewall Paradox~\cite{HarHay13}.

This interplay is also fundamentally important for many cryptographic applications. For example, despite certain information-theoretically secure quantum protocols~\cite{BenBra84}, most cryptographic tasks still require computational constraints even when using quantum information~\cite{Mayers96,LoChau98}. Nevertheless, combining quantum information with computational constraints opens up numerous possibilities, from minimizing computational assumptions~\cite{GLSV21,BCDM21} to classically-impossible applications~\cite{Aaronson11}.

The previous examples show that scenarios with quantum information can be fundamentally altered by the presence of computational considerations. 
It is therefore important to develop a broad understanding of quantum information in the computationally bounded setting. This includes the famous no-go theorems of quantum information. Numerous prior works have studied no-cloning in the computational setting (see references in Section~\ref{sec:related}). However, the computational difficulty of telegraphing has, to the best of our knowledge, not been previously studied. As our work shows, the equivalence of two of the most important quantum no-go theorems no longer holds in the computationally bounded setting, giving a very different picture and allowing for new possibilities that do not exist in the information-theoretic setting.

\paragraph{Cryptographic Applications.} \label{par:crypto-applications}
Besides addressing fundamental questions, we also explore potential cryptographic applications of our separation.

Concretely, consider the following key exfiltration scenario: a server contains a cryptographic key, say, for decrypting messages. An attacker remotely compromises the server, and then attempts to exfiltrate the key, sending it from the compromised server to the attacker's computer.

A classical approach to mitigate this problem is \emph{big key cryptography}~\cite{C:BelKanRog16,CCS:BelDai17,C:MorWic20}, where the secret decryption key is made inconveniently large. This may make it impossible for the attacker to exfiltrate the key (perhaps there are bandwidth constraints on outgoing messages from the server) or at least makes it easier to detect the exfiltration
(the large communication of the key may be easily detected). Unfortunately, such large keys are also inconvenient for the honest server, as now the server needs significant storage for each key (perhaps the server is storing keys for many different users). Moreover, decrypting using the key may be problematic, since the server will have to compute on a large key, which at least requires reading it from storage. If the server is decrypting many messages simultaneously using parallelism, then each process would presumably need to separately load the entire key from memory.

A quantum proposal would be to have decryption keys be quantum states. It is still reasonable to consider such a setting where all communication is classical: after all, the messages being encrypted and decrypted may just be classical. The server could therefore force all outgoing communication to be classical by measuring it. This would prevent the remote adversary from exfiltrating the key, by the non-telegraphability of the key. 

Since telegraphing trivially implies cloning, we note that any classical program which has been quantum copy protected~\cite{Aaronson11} will be immune from classical exfiltration. Copy protection for decryption keys was first considered by~\cite{EPRINT:GeoZha20}, and was constructed from indistinguishability obfuscation by~\cite{ColLiuLiuZha21}, along with copy protection for pseudorandom functions and signatures.

However, using copy protection comes with its own limitations. Indeed, suppose the server is decrypting a large volume of incoming communication under a single decryption key. Classically, the server could divide the communication across several processors, with each decrypting in parallel. Unfortunately, this requires giving each processor a copy of the key. While trivial classically, the whole point of copy protection is to prevent copying. In fact,~\cite{oneshot} consider exactly the task of preventing the use of parallelism via copy protection. The server could simply store numerous copies of each copy-protected key, but it would have to store these keys forever, even when the server is sitting idle or processing other tasks. This could be a major burden on the server. It also requires security to hold given multiple copies of the program, a non-trivial extension to single-copy security~\cite{TCC:LLQZ22}.

Instead, we imagine a protocol where the quantum keys \emph{are} copy-able, but remain impossible to telegraph. This would protect against exfiltration, while allowing the server to only store a single copy of the key for long-term use. Then, if the incoming communication load ever becomes large, it can copy the key and spread the copies amongst several quantum processors that process the communication in parallel. After the load subsides and processors would return to being idle, the copies of the key can simply be discarded.

Assuming states that can be cloned but not telegraphed, we show how to realize an encryption scheme with the above features:
\begin{theorem}[Informal presentation of Theorem~\ref{thm:non-exfiltration-sufficient-2}]\label{thm:exfilinf} Assume the existence of clonable-but-untelegraphable states which can be efficiently sampled. Additionally assume the existence of \emph{extractable witness encryption for $\QMA$}. Then there exists public key encryption with quantum secret keys that can be cloned but not exfiltrated.
\end{theorem}
For the necessary witness encryption, we could use~\cite{ITCS:BarMal22} as a candidate. Note that the states we construct relative to an oracle in Theorem~\ref{thm:maininf} are efficiently sampleable. However, witness encryption requires non-black box use of the $\QMA$ language, meaning it cannot be applied to query-aided languages like that implied by Theorem~\ref{thm:maininf}. However, any standard-model realization of clonable-but-non-exfiltrateable states would suffice, and our Theorem~\ref{thm:maininf} gives some evidence that such states exist.

This is just one potential application of no-telegraphing that does not follow immediately from no-cloning. Our hope is that this work will motivate further study of no-telegraphing in cryptography.

\paragraph{On Oracles.} Our separation between no-cloning and no-telegraphing requires oracles. Given the current state of complexity theory and the fact that these no-go properties are equivalent for computationally unbounded adversaries, we cannot hope to achieve unconditional separations between them in the standard model. As such, either computational assumptions or a relativized separation (that is, oracles) are required.

For cryptographic applications such as Theorem~\ref{thm:exfilinf}, certainly a standard-model construction from computational assumptions would be needed. On the other hand, by using oracles, we are able to give an unconditional separation, independent of what assumptions may or may not hold. While such a relativized separation does not necessarily rule out a standard-model equivalence, it shows a fundamental barrier to such an equivalence. Indeed, an immediate corollary of Theorem~\ref{thm:maininf} is:
\begin{corollary}There is no black box reduction showing the equivalence of cloning and telegraphing in the computational setting.
\end{corollary}

We also note that our oracles as stated are sampled from a distribution, rather than being fixed oracles. This is typical of the cryptographic black-box separation literature. In the setting of uniform adversaries, a routine argument allows us to turn this into a fixed oracle relative to which the separations hold. 
We do this explicitly in the proof of Theorem~\ref{thm:complexityinf} to get a separation relative to a fixed unitary oracle, and we further note that this directly implies such a separation between cloning and telegraphing as well.

\subsection{Other Related Work}\label{sec:related}

Cloning in the complexity-theoretic setting has been extensively studied during the last decade, in the context of public key quantum money~\cite{Aaronson11,knots,subspaces,lightning,classicalbank} and copy protection~\cite{Aaronson11,ssl,ColLiuLiuZha21}. 

A recent development in quantum money has been quantum money with \emph{classical} communication~\cite{semiquantum,oneshot,classicalbank}. This can be seen as a complement to our separation, giving a setting where a quantum state \emph{is} telegraphable, but not clonable. In order to overcome the trivial telegraphing-implies-cloning result, however, these works move to \emph{interactive} telegraphing, involving two or more messages between sender and receiver. Moreover, telegraphing happens in only a weak sense: the receiver does not get the original quantum state. Instead, the sender's quantum money state is actually irreversibly destroyed, but in the process the receiver is able to create a single new quantum money state.

\subsection{Technical Overview}

Let $f$ be a random function with codomain much smaller than domain. Our clonable-but-not-telegraphable states will be the the superpositions over pre-images of $f$:
$$|\psi_{z}\rangle = \frac{1}{\sqrt{|f^{-1}(z)|}} \sum_{x | f(x) = z} |x\rangle$$
where $f^{-1}(z) := \{x | f(x) = z\}$ is the set of preimages of $z$ in $f$.

As of now, the $|\psi_z\rangle$ are easily shown to be \emph{un}clonable: if one could create two copies of $|\psi_z\rangle$, then measuring each would give two pre-images $x_1,x_2$ such that $f(x_1)=z=f(x_2)$. Since $f$ has a small codomain, there are exponentially many $x$ in the support of $|\psi_z\rangle$, and therefore $x_1\neq x_2$ with overwhelming probability. Thus we obtain a collision for $f$, which is known to be intractable for query-bounded algorithms to random oracles, even ones with small codomains~\cite{CIS:Zhandry15}. 

That the $|\psi_z\rangle$ are unclonable seems to be counterproductive for our aims. But it allows us to also readily prove that the $|\psi_z\rangle$ are also un-telegraphable: if one could telegraph $|\psi_z\rangle$, it means one can generate a classical $a_z$ such that from $a_z$ it is possible to reconstruct $|\psi_z\rangle$. But by running reconstruction multiple times, one obtains multiple copies of $|\psi_z\rangle$, contradicting no-cloning. This is not exactly how we prove un-telegraphability, but provides an intuition for why it should be true.

Now that we have an untelegraphable set of states, we make them clonable by adding a cloning oracle, which very roughly maps
\[|\psi_z\rangle\mapsto|\psi_z\rangle|\psi_z\rangle\]
for all valid states $|\psi_z\rangle$ and does nothing on states that are not uniform superpositions of pre-images. This clearly makes the $|\psi_z\rangle$ clonable. The challenge is then to prove that telegraphing is still impossible, even given this cloning oracle. This is proved through a sequence of stages:
\begin{itemize}
    \item {\bf Stage 1.} Here, we remove the cloning oracle, and just consider the oracle $f$. We show that, with arbitrary classical advice $a_z$ of polynomially-bounded size dependent on $z$ (which could have been constructed in an arbitrary inefficient manner), it is impossible for a query-bounded algorithm to reconstruct $|\psi_z\rangle$. This is proved by showing that such a reconstruction procedure could be used to contradict known lower bounds on the hardness of finding $K$ collisions~\cite{Hamoudi2020-mz}. 
    
    The above shows that even if we give the \emph{sender} the cloning oracle, then telegraphing is still impossible for a query-bounded receiver, as long as the receiver does not have access to the cloning oracle. Indeed, the hypothetical output of such a sender would be an $a_z$ contradicting the above.
    \item {\bf Stage 2.} Here, we upgrade the receiver to have a limited cloning oracle that only clones a single $|\psi_z\rangle$, namely the unique state $|\psi_z\rangle$ that the receiver is trying to reconstruct.
    
    The intuition is that such a limited cloning oracle is of no use, since in order to query it on a useful input, the receiver needs to have $|\psi_z\rangle$ in the first place. We make this formal using a careful analysis.
    
    \item {\bf Stage 3.} Finally, we give the receiver the full cloning oracle. We show that if there is such a query-bounded receiver that can successfully reconstruct, then we can compile it into a query-bounded receiver for Stage 2, reaching a contradiction.
    
    This is the most technically challenging part of our proof. The rough idea is that the Stage 2 receiver will simulate a set of imposter oracles, where it forwards queries relating to $z$ to its own $z$-only cloning oracle, and all other queries it handles for itself. This simulation is not perfect, and care is required to prove that the simulation still allows for successful reconstruction.
\end{itemize}

Putting these together, we prove Theorem~\ref{thm:maininf}, that there cannot exist \emph{any} telegraphing scheme for the set of $|\psi_z\rangle$ with a query-bounded receiver (regardless of whether the sender is query bounded).

\begin{remark}The above description requires two oracles, a classical random oracle (queryable in superposition) and the cloning oracle. We first note that superposition access to a random oracle is in particular unitary, so the classical random oracle is also a unitary. Second, we can view these two oracles as a single quantum oracle, which operates on two sets of registers, applying one oracle to one register and the other oracle to the other. For the single combined oracle to be equivalent to the two individual oracles, we only need that the individual oracles have efficiently constructible fixed points. This is true of both the oracles we use. Thus, we obtain a separation relative to a single oracle sampled from an appropriate distribution.
\end{remark}

\subsection{Open Problems and Future Directions}

\paragraph{Variations on Telegraphing and Cloning.}
In this work, we consider telegraphing to occur in a single directed round.
We can see, however, as mentioned above, that there are settings in which allowing multiple rounds of classical interaction provides new capabilities. 
An open problem would thus be to show that our scheme resists telegraphing even by adversaries that are allowed multiple rounds of interaction, or alternatively to give another clonable scheme which does.
Such a result would suggest that under computational constraints, interactive telegraphing is a completely independent property that neither implies nor is implied by cloning.
% Related question: What is the round complexity of telegraphing? Are there states which can be interactively telegraphed but not a polynomial number of rounds? This is a relevant question even for unbounded adversaries.

Along similar lines, one can also define extended versions of cloning and telegraphing in which the adversaries are given multiple copies of the state to clone or telegraph. They would then be tasked with producing one additional copy in the case of cloning, or telegraphing one instance across the classical channel. These extended tasks appear to be inherently easier, as the multiple copies may allow for partial tomography of the state. At the very least, they are no more difficult. What kinds of states are resistant to these extended tasks under computational constraints, and what are the implication relationships between them?

% % Formal definitions of the extended tasks

% \begin{definition}[Interactive telegraphing]
% \label{def:interactive_telegraphing}
% A scheme is said to be $\varepsilon$-\textbf{worst case $k$-round interactively telegraphable} if there exists a classically interacting pair of quantum algorithms
% $\mathsf{Send}(|\psi\rangle)$
% and
% $\mathsf{Receive}()$
% such that when $\mathsf{Send}$ (the sender) is given a quantum state $|\psi\rangle \in \mathcal H_z$, the sender and receiver interact for at most $k$ rounds and then the receiver $\mathsf{Receive}$ outputs a quantum state $|\phi\rangle$ that passes verification for $z$ with probability at least $\varepsilon$.

% Similarly, a scheme is said to be $\varepsilon$-\textbf{average case $k$-round interactively telegraphable} if there exist $\mathsf{Send}$ and $\mathsf{Receive}$ which succeed in this way for 
% $(z, |\psi_z\rangle) \leftarrow \mathsf{Gen}()$ with probability at least $\varepsilon$.
% \end{definition}

% \begin{definition}[$k$-Cloning]
% \label{def:k-cloning}
% \end{definition}

% \begin{definition}[$k$-Telegraphing]
% \label{def:k-telegraphing}
% \end{definition}

% \begin{definition}[$k$-Deconstruction]
% \label{def:k-deconstruction}
% \end{definition}

% \paragraph{Operational Bounded and Operational Statistical Settings.}
\paragraph{Cloning and Telegraphing the Functionality of a State.}
For some applications, the functionality of quantum states is more important than the states themselves. For example, if several different cryptographic keys all allow performing the same decryption, then it may be important to guarantee that no procedure that is given one such key is able to produce another valid key, even if it is not identical to the original.
One way to capture this is by replacing the individual states in the collection $S$ (of states to be cloned or telegraphed) with sets of states that all have equivalent functionality, or alternatively, by the subspaces spanned by such sets.
Such a definition in particular would be general enough to include the schemes for quantum money with classical communication mentioned in Section~\ref{sec:related}.
In particular, single states are trivial singleton sets for functionality, so any cloning or telegraphing scheme for states is also a scheme for functionality, 
which means that our separation result also holds in this context. 
Interestingly, however, considering the functionality of states rather than the states themselves has the potential to allow new or stronger separations that do not exist for tasks on states.
% For instance it is not a priori clear that the ability to clone a subspace with some success probability implies the ability to telegraph it with comparable probability, \emph{even for an unbounded adversary}.

\paragraph{No-go Properties Exclusive to Efficient Computation.}
An artifact of our proof of the separation between cloning and telegraphing is that we actually prove a stronger separation between cloning and a simplified task that we call \emph{reconstruction}. 
This motivates giving this task its own no-go property: 
A set of quantum states $S$ is computationally reconstructible if there is a polynomial-time reconstruction procedure that constructs a state in $S$ from a classical description that uniquely identifies it from within the set. 
Note, however, that this task is trivial for unbounded procedures: 
Any quantum state on a finite Hilbert space can be uniquely identified with its list of complex amplitudes, up to some precision error, which means that an unbounded procedure can always reconstruct it. 
This means that unlike no-cloning and no-telegraphing, no-reconstruction exists exclusively in the context of efficient computation. What other quantum no-go properties have until now been overlooked because they do not appear until computational efficiency is considered?

\paragraph{Separation by a Classical Oracle.}
The separations presented in this paper, between efficient cloning and efficient telegraphing, as well as between $\clonableQMA$ and $\QCMA$, both use quantum oracles. 
This is very natural for the cloning vs. telegraphing problem, as the cloning oracle we use is precisely the black-box cloner, from which we show that no efficient telegraphing protocol can be constructed.
On the other hand, for the separation between $\clonableQMA$ and $\QCMA$, it is reasonable to wonder whether a fully quantum oracle is necessary. That is, is there a classical oracle, accessible in superposition, relative to which the separation between $\clonableQMA$ and $\QCMA$ still holds? Such an oracle would as an immediate consequence also separate $\QCMA$ from $\QMA$, an open problem that remains unresolved for standard classical oracles despite much recent progress.

% For the cryptographic applications, we present a proof-of-concept construction of parallelizable but non-exfiltratable encryption from a pair of strong assumptions, so an open problem is to base it off of weaker assumptions.

% While our results give the first demonstration of the existence of 
% % a separation between cloning and telegraphing, 
% states that are efficiently clonable but not efficiently telegraphable,
% they do not readily imply a method for creating such states in the real world. 
% Doing so would require replacing the oracles in our scheme with realistic functions and cryptographic primitives, 
% whose security would have to be built on plausible computational assumptions.
% A related direction would be to use such a scheme to instantiate concrete cryptographic applications, 
% including preventing key exfiltration, 
% as mentioned in Section~\ref{par:crypto-applications}.

\subsection{Paper Outline}
We start with some preliminaries in Section~\ref{sec:prelims}. In Section~\ref{sec:no-go-tasks}, we formally define the no-go tasks of cloning and telegraphing, as well as the related task of reconstruction. Section~\ref{sec:cloning-without-telegraphability} forms the main technical part of the paper, and includes the formal statement and proof of Theorem~\ref{thm:maininf}. The complexity theoretic applications of clonable-but-untelegraphable states appear in Section~\ref{sec:complexity}, where we define the complexity class $\clonableQMA$, and give the proof of Theorem~\ref{thm:complexityinf}. 
We end with applications to cryptography in Section~\ref{sec:cryptography}, with a formalization of the notion of parallelizable but non-exfiltratable encryption and the formal statement and proof of Theorem~\ref{thm:exfilinf}.

\section{Preliminaries}\label{sec:prelims}

\subsection{Quantum Computation and Computational Efficiency}

As this paper touches on fundamental properties of quantum mechanics and their intersection with computational efficiency, it is geared towards a wide audience of physicists and computer scientists. However, due to limited space, we only give a basic primer on quantum information and refer the reader to the textbook by Nielsen and Chuang~\cite{Nielsen2012-fh} for a more complete background.

A quantum state $|\psi\rangle$ is a normalized vector in a Hilbert space $\mathcal{H}$. If we choose a basis $\{|\phi_i\rangle\}_i$ for $\mathcal{H}$, then we can measure $|\psi\rangle$ in that basis, at which point we get outcome $i$ with probability $| \langle \phi_i | \psi \rangle |^2$. A distribution $\{p_i\}_i$ over a set of quantum states $|\psi_i\rangle$, often called a mixed state, can be represented as a density matrix $\rho$, which has the form $\rho = \sum_i p_i | \psi_i \rangle\langle \psi_i |$, and which captures the randomness of the quantum state and anything that can be measured about it. If we measure $\rho$ in the basis $\{|\phi_i\rangle\}_i$, we get outcome $i$ with probability $\langle\phi_i | \rho | \phi_i\rangle = \sum_j p_j | \langle \phi_i | \psi_j \rangle|^2$. We can course-grain this measurement by employing a projection-valued measurement (PVM), which groups together the projectors $\ket{\phi_i}\bra{\phi_i}$ into measurement operators $\Pi_k = \sum_{i \in S_k} \ket{\phi_i}\bra{\phi_i}$, where $\sum_k \Pi_k = \mathbb I$ and $\mathbb I$ is the identity operator on $\mathcal{H}$. If we measure $\rho$ with the PVM described by $\{\Pi_k\}_k$, we get outcome $k$ with probability $\tr(\Pi_k \rho)$. We say that $\{\Pi_k\}_k$ is a binary projective measurement if has only two potential outcomes.

Quantum states are transformed to other quantum states via unitary transformations. A unitary $\mathcal{U}$ transforms a pure quantum state $|\psi\rangle$ by sending it to $\mathcal{U}|\psi\rangle$, and it transforms a mixed state $\rho$ by sending it to $\mathcal{U} \rho \mathcal{U}^{\dagger}$.
A quantum algorithm is a sequence of basic unitary transformations $\mathcal{U}_1, \cdots, \mathcal{U}_t$ chosen from some set $U$. 
It acts on a pure quantum state $|\psi\rangle$ by applying each one in turn as $\mathcal{U}_t \mathcal{U}_{t-1} \cdots \mathcal{U}_1|\psi\rangle$ (and on a mixed state analogously). If an algorithm performs $t$ such unitaries, then we say that the algorithm runs in time $t$.
Note that while non-unitary transformations such as measurement or randomness can occur in general, for our purposes, we can always assume without loss of generality that all such non-unitary behavior occurs in a single measurement at the end of the algorithm.

A quantum oracle is a predefined unitary transformation that is applied atomically in a black-box fashion. That is, it is performed as if it were a single operation which cannot be broken down and does not need to have a computationally feasible implementation. Quantum oracles can in particular also implement a classical function $f: \mathcal{X} \rightarrow \mathcal{Y}$ by applying the unitary that transforms each $|x\rangle|y\rangle$ where $x \in \mathcal{X}$ and $y \in \mathcal{Y}$ to $|x\rangle|y \oplus f(x)\rangle$. In this case, we refer to the oracle, and to the function $f$ itself, as a classical oracle, despite it operating on data that may be quantum.
Each time that an oracle is applied is called a query to that oracle.
A quantum oracle algorithm is an algorithm that may at any point perform a query to any oracle from a predefined collection of oracles by applying the unitary for that oracle onto some subset of its registers. If this occurs $q$ times, we say that it is a $q$-query algorithm, or that it runs in $q$ queries.

When considering the runtime of an algorithm, it is often useful to ignore lower order influences and focus its asymptotic behavior as the problem grows. All asymptotics are in terms of the size of the problem given as input, or a security parameter, $n \in \mathbb N$. In asymptotic notation, $O(f(n))$ refers to the set of functions that are bounded from above by $c \cdot f(n)$ for some constant $c$ and sufficiently large $n$. The notation $\Omega(f(n))$ likewise refers to functions that are bounded from below by $c \cdot f(n)$ for some positive constant $c$ and sufficiently large $n$. We say that a function $g(n)$ is polynomial in $n$, notated $\poly(n)$, if it is $O(n^c)$ for some constant real number $c$. We say that $g(n)$ is negligible, notated $\negl(n)$ if for every constant real number $c$, it is not $\Omega(n^{-c})$. On the other hand, we say that $g(n)$ is non-negligible if it is not negligible, and we say that it is overwhelming if $0 \le g(n) \le 1$ and $1 - g(n)$ is negligible. We say that an event $E$ happens with high probability when $\Pr[E] \ge \frac12 + \Omega(1)$.

We say that an algorithm is efficient or polynomial-time, if it runs it time at most $t(n)$ and $t(n)$ is polynomial in $n$. We say that an oracle algorithm is query-efficient or polynomial-query (or sometimes just polynomial-time) if it makes at most $q(n)$ oracle queries and $q(n)$ is polynomial in~$n$.

\subsection{Decision Problems}

A promise problem is a pair of disjoint sets of strings, $\mathcal{L} = (\mathcal{L}_{\YES}, \mathcal{L}_{\NO})$, where $\mathcal{L}_{\YES}, \mathcal{L}_{\NO} \subseteq \bits^*$, $\mathcal{L}_{\YES} \cap \mathcal{L}_{\NO} = \emptyset$. Strings in $\mathcal{L}_{\YES}$ and $\mathcal{L}_{\NO}$ are called $\YES$ instances and $\NO$ instances, respectively.
A language (sometimes called a decision problem) has the extra guarantee that $\mathcal{L}_{\YES} \cup \mathcal{L}_{\NO} = \bits^*$, in which case it may be specified by only giving $\mathcal{L}_{\YES}$. When otherwise clear from context, or when the distinction is not meaningful, we use the term decision problem to refer to either languages or promise problems, and use the more precise term where the distinction matters.
An oracle decision problem, or a black-box decision problem, is defined similarly, but where $\mathcal{L}_\YES$ and $\mathcal{L}_\NO$ are instead sets of oracles, and an instance is an oracle in one of the two sets.
For oracle decision problems, the size of the instance is taken to be the size of the input to the oracle. 
An average-case decision problem, $(\mathcal{L}, \mathcal{D})$, is a decision problem $\mathcal{L}$, paired with a probability distribution $\mathcal{D}$ over its instances. We normally require that $\mathcal{D}$ can be sampled by a polynomial-time quantum algorithm.

We say that an algorithm decides or computes a decision problem if on input $x \in \mathcal{L}_{\YES} \cup \mathcal{L}_{\NO}$, it accepts (outputs $\YES$ or $1$) whenever $x \in \mathcal{L}_{\YES}$, and rejects (outputs $\NO$ or $0$) whenever $x \in \mathcal{L}_{\NO}$. 
An algorithm for an average-case problem may fail on any subset of its instances over which the distribution, $\mathcal{D}$, assigns a combined low probability. %of at most $\frac12 - \varepsilon$.
A complexity class is a set of (generalized) decision problems. There may be corresponding versions of the same complexity class for promise problems, languages, oracle problems, or average-case problems, and which one is meant is usually made clear from context, but whenever this is not the case, we specify explicitly which one we mean. 
We specifically say that an average-case decision problem $(\mathcal{L}, \mathcal{D})$ is hard for a certain complexity class if for every algorithm, there exists a negligible function $\varepsilon$ such that the algorithm fails to satisfy the conditions of the class with probability at least $\frac12 - \varepsilon$ over the distribution $\mathcal{D}$ on the instances.

\subsection{Complexity Classes}

In classical complexity theory, the class $\NP$ is the class of decision problems that have efficiently verifiable proofs (called \emph{witnesses}). 
That is, such problems have a polynomial time classical verifier, for which given any $\YES$ instance of the decision problem, there is a polynomial length classical witness string that allows the verifier to verify the instance as a $\YES$ instance, while at the same time no purported witness would lead to verifying a $\NO$ instance. The related complexity class $\MA$ allows the verifier to use randomness and bounded error. We concern ourselves mainly with the quantum generalizations of these classes, known as $\QCMA$ and $\QMA$, which allow the verifier to be a polynomial-time quantum Turing machine, and the witnesses to be classical strings or quantum states, respectively~\cite{aharonov2002quantum,892141}.

\begin{definition}[$\QCMA$]\label{def:QCMA}
A decision problem $\mathcal{L} = (\mathcal{L}_{\YES}, \mathcal{L}_{\NO})$ is in $\QCMA(c, s)$ if there exists a polynomial time quantum verifier $V$, and a polynomial $p$, such that
\begin{itemize}
    \item \textbf{Completeness:} if $x \in \mathcal{L}_{\YES}$, then there exists a classical witness $w \in \bits^{p(|x|)}$ such that $V$ accepts on input $\ket x \ket w$ with probability at least $c$
    \item \textbf{Soundness:} if $x \in \mathcal{L}_{\NO}$, then for all classical strings $w^* \in \bits^{p(|x|)}$, $V$ accepts on input $\ket{x}\ket{w^*}$ with probability at most $s$.
\end{itemize}
\end{definition}

\begin{definition}[$\QMA$]\label{def:QMA}
A decision problem $\mathcal{L} = (\mathcal{L}_{\YES}, \mathcal{L}_{\NO})$ is in $\QMA(c, s)$ if there exists a polynomial time quantum verifier $V$, and a polynomial $p$, such that
\begin{itemize}
    \item \textbf{Completeness:} if $x \in \mathcal{L}_{\YES}$, then there exists a quantum witness $\ket \psi$ on $p(|x|)$ qubits such that $V$ accepts on input $\ket x \ket \psi$ with probability at least $c$
    \item \textbf{Soundness:} if $x \in \mathcal{L}_{\NO}$, then for all quantum states $\ket{\psi^*}$ on $p(|x|)$ qubits, $V$ accepts on input $\ket{x}\ket{\psi^*}$ with probability at most $s$.
\end{itemize}
\end{definition}
We take $\QCMA = \QCMA(\frac9{10}, \frac1{10})$ and $\QMA = \QMA(\frac9{10}, \frac1{10})$, although by making use of parallel repetition for error reduction, the completeness and soundness parameters can be set arbitrarily, with a wide leeway in the size of the completeness-soundness gap, without changing the class~\cite{aharonov2002quantum}.%
\footnote{The standard practice is usually to set the completeness and soundness parameters to $2/3$ and $1/3$ respectively. Since, as mentioned, the exact values are arbitrary, we find it convenient instead to set them to $9/10$ and $1/10$, in order to simplify the presentation of some proofs which would otherwise require more frequent invocations of parallel repetition to increase the gap.}

% \paragraph{Terms still to define}
% \begin{itemize}
% \end{itemize}

% \subsection{Terms}
% Need to define the following terms
% \begin{itemize}
%     \item classical oracle
%     \item quantum oracle
%     \item unitary application of a classical oracle
%     \item computational efficiency and asymptotics
%     \item efficiency vs efficiency in terms of oracles
%     \begin{itemize}
%         \item computationally efficient algorithm
%         \item efficient oracle algorithm / query-efficient algorithm
%         \item unbounded algorithms
%     \end{itemize}
%     \item negligible, non-negligible, overwhelming, etc.
%     \item images, preimages, etc  <--- haven't defined
%     \item challenge vs target  <--- haven't defined
%     \item (possibly) quantum teleportation   <--- haven't defined
% \end{itemize}

\subsection{Query Magnitudes and Modifying Oracles}
When working with quantum oracle algorithms, 
it is often useful to be able to bound the effect that replacing one oracle with another can have on the result of the computation.
To this end, we recall the following definition and two theorems due to Bennett, Bernstein, Brassard, and Vazirani~\cite{Bennett1997-xg}:
\begin{theorem}[Theorem 3.1 from~\cite{Bennett1997-xg}] \label{thm:Ben+97_3.1}
If two unit-length superpositions are within Euclidean distance $\varepsilon$ then observing the two superpositions gives samples from distributions which are within total variation distance at most $4\varepsilon$.
\end{theorem}

\begin{definition}[Definition 3.2 from~\cite{Bennett1997-xg}]
\label{def:Ben+97_3.2}
Let $|\phi_i\rangle$ be the superposition of $M^A$ on input $x$ at time $i$. We denote by
$q_y(|\phi_i\rangle)$ the sum of squared magnitudes in $|\phi_i\rangle$ of configurations of $M$ which are querying the
oracle on string $y$. We refer to $q_y(|\phi_i\rangle)$ as the query magnitude of $y$ in $|\phi_i\rangle$.
\end{definition}

\begin{theorem}[Theorem 3.3 from~\cite{Bennett1997-xg}] \label{thm:Ben+97_3.3}
Let $|\phi_i\rangle$ be the superposition of $M^A$ on input $x$ at time $i$. Let $\varepsilon > 0$.
Let $F \subseteq [0, T - 1] \times \Sigma^*$
be a set of time-strings pairs such that $\sum_{(i,y)\in F} q_y(|\phi_i\rangle) \le \frac{\varepsilon^2}{T}$.
Now suppose the answer to each query $(i, y) \in F$ is modified to some arbitrary fixed $a_{i,y}$ (these answers need not be consistent with an oracle). Let $|\phi'_i\rangle$ be the time $i$ superposition of
$M$ on input $x$ with oracle $A$ modified as stated above. Then $\big||\phi_T\rangle - |\phi'_T\rangle\big| \le \varepsilon$.
\end{theorem}

\section{Fundamental Tasks and Their No-go Properties}
\label{sec:no-go-tasks}

\subsection{Schemes of Quantum States}
We introduce the following syntax for a scheme of quantum states. A scheme is the basic structure on which the quantum no-go properties may or may not apply. In other words, some schemes may be clonable, for instance, while other schemes may not. Schemes consist primarily of a collection of quantum states, but they can also specify the collection of oracles which may be used, as well as a distribution for sampling from those states.

% Quantum information definition
\begin{definition}[Scheme]
\label{def:scheme} 
In the context of quantum no-go properties, a \textbf{scheme}, $(S, \mathcal{D}, \mathcal{O})$, is an indexed collection of quantum states 
$S = \{|\psi_i\rangle\}_{i \in \mathcal{Z}}$ over an index set $\mathcal{Z}$ (which we call the set of labels), a distribution $\mathcal{D}$ over the labels, and a collection $\mathcal{O}$ of any quantum oracles that may be used.
\end{definition}

Whenever either the distribution or the oracles are irrelevant or otherwise clear from context, we will drop them from the notation and write $(S, \mathcal{O})$, $(S, \mathcal{D})$, or simply $S$.
Note that the distribution $\mathcal{D}$, which allows sampling from the collection of states, is only important for defining average-case security of the scheme, and $\mathcal{O}$ is only necessary when considering oracle algorithms.

Under a certain scheme $(S, \mathcal{D}, \mathcal{O})$, \emph{verification} of an unknown quantum state $|\phi\rangle$ for a label $z$ is the measurement of whether $|\phi\rangle$ passes for the intended state $|\psi_z\rangle$, which succeeds with probability $p = |\langle \psi_z | \phi \rangle|^2$. 
When we say that an algorithm succeeds in passing verification with some probability $p$, we mean that verification succeeds with that probability over the randomness of the algorithm's output as well as the randomness of the sampling from $\mathcal{D}$ and that of the verification measurement. 
That is, if the algorithm is randomized and outputs a mixed state $\rho_z$ when label $z$ is drawn, then we say that it succeeds at verification for $z$ with probability $p = \mathbb E_{z \leftarrow \mathcal{D}} \langle \psi_z | \rho_z | \psi_z \rangle$. 
This is the expected fidelity of the states produced by the algorithm with the intended state.
Whenever an algorithm is tasked with passing verification for a label $z$, we call $z$ the target label and we call $|\psi_z\rangle$ the target state.

% % Cryptographic definition
% All schemes have a verification function $\mathsf{Ver}: \mathcal{Z}\times \mathcal{H} \rightarrow \{0,1\}$, which performs a projective measurement,
% $\mathsf{Ver}(z,|\psi\rangle)$, on a classical label or key $z$, which codes for some subspace $\mathcal H_z$, and measures whether $|\psi\rangle$ is in that subspace. If verification passes, it also returns the resulting state as well (that is, it returns the normalized projection of $|\psi\rangle$ into $\mathcal H_z$).

% Schemes may also have a randomized function $\mathsf{Gen}: \rightarrow \mathcal{Z}\times \mathcal{H}$, which allows jointly sampling a pair $(z, |\psi_z\rangle)$ that passes verification. This function is not necessary for the worst-case security definitions. It is, however, important for defining average-case security. Furthermore, its existence demonstrates that an initial input state can be efficiently sampled.

\subsection{Cloning and Telegraphing}

We now formally define the tasks of cloning and telegraphing.

\begin{definition}[Cloning]
\label{def:cloning}
A scheme $S$ is said to be $\eta$-\textbf{worst case clonable} if there exists a quantum algorithm 
$\mathsf{Clone}(|\psi\rangle)$
such that for every label $z \in \mathcal{Z}$, when given $|\psi_z\rangle$, its corresponding quantum state in $S$, 
% $\in \mathcal H_z$ (that is, a state verified by $\mathsf{Ver}(z, \cdot)$),
returns a quantum state $|\phi\rangle$ on two registers that, with probability at least $\eta$, passes verification for $z$ on both registers simultaneously. That is, $\left|\big(\,\bra{\psi_z} \otimes \bra{\psi_z} \,\big) \ket{\phi}\right|^2 \ge \eta$.

$(S, \mathcal{D})$ is said to be $\eta$-\textbf{average case clonable} if there exists a quantum algorithm 
$\mathsf{Clone}(|\psi\rangle)$
that succeeds at the cloning task with probability $\eta$ when $z$ is sampled from the distribution $\mathcal{D}$.
% % $(z, |\psi_z\rangle) \leftarrow \mathsf{Gen}()$, 
% returns a quantum state on two registers, which with probability at least $\eta$ passes verification for $z$ on both registers simultaneously.
\end{definition}

% Note that the output of $\mathsf{Clone}(|\psi\rangle)$ doesn't necessarily have to be two exact copies of the original state. It may be any pair of states that pass verification, or even a superposition of such pairs.

\begin{definition}[Telegraphing]
\label{def:telegraphing}
A scheme $S$ is said to be $\eta$-\textbf{worst case telegraphable} if there exists a pair of quantum algorithms
$\mathsf{Send}(|\psi\rangle) \rightarrow c$
and 
$\mathsf{Receive}(c) \rightarrow |\phi\rangle$ where $c$ is a classical string,
such that for every label $z \in \mathcal{Z}$, when given $|\psi_z\rangle$, its corresponding quantum state in $S$,
$|\phi\rangle := \mathsf{Receive}(\mathsf{Send}(|\psi_z\rangle))$ passes verification for $z$ with probability at least $\eta$.

$(S, \mathcal{D})$ is said to be $\eta$-\textbf{average case telegraphable} if there exists a pair of quantum algorithms
$\mathsf{Send}(|\psi\rangle) \rightarrow c$
and  
$\mathsf{Receive}(c) \rightarrow |\phi\rangle$
that succeed at the telegraphing task with probability $\eta$ when $z$ is sampled from the distribution $\mathcal{D}$.

% Similarly, a scheme is said to be $\eta$-\textbf{average case telegraphable} if there exists a pair of quantum algorithms
% $\mathsf{Send}(|\psi\rangle) \rightarrow c$
% and  
% $\mathsf{Receive}(c) \rightarrow |\phi\rangle$
% such that for 
% $(z, |\psi_z\rangle) \leftarrow \mathsf{Gen}()$,
% $|\phi\rangle := \mathsf{Receive}(\mathsf{Send}(|\psi\rangle))$ passes verification for $z$ with probability at least $\eta$.
\end{definition}

Note that quantum teleportation is 
the process by which a quantum state can be transmitted through a classical channel by the use of pre-shared quantum entanglement~\cite{PhysRevLett.70.1895}.
Telegraphing can thus be viewed as describing a quantum teleportation protocol without the use of entanglement: $\mathsf{Send}$ converts the quantum state $|\psi_z\rangle$ to a classical description $c$, which $\mathsf{Receive}$ then converts back into $|\psi_z\rangle$, or an approximation thereof. 
This is why the no-go theorem of the telegraphing task for general quantum states is often referred to as the 
\emph{no-teleportation theorem},
% ``no-teleportation theorem'', 
a name first coined by the originator of the theorem~\cite{Werner_1998}. 
This terminology can be confusing, however, since teleportation \emph{is} in fact always possible when the sending and receiving parties are allowed to start out with an additional entangled quantum state.
To sidestep this confusion, throughout this paper we instead use the term \emph{telegraphing} for the unentangled no-go task.
Here, and throughout this paper, any pair of algorithms attempting to achieve the telegraphing task are attempting to do so without the use of pre-shared entanglement.

\subsection{Information Theoretic No-go Theorems}

% \paragraph{The No-go Theorems.}
We now state a version of the (information theoretic) no-go theorems for these two tasks. The No-Cloning Theorem was first proved by three independent papers~\cite{Park70,WoottersZurek82,Dieks82}, but the version we present here is due to~\cite{YUEN1986405}. The No-Telegraphing Theorem (originally called the No-Teleportation Theorem), a corollary of the No-Cloning Theorem, is due to~\cite{Werner_1998}. We present the two theorems here together to emphasize the direct connection between~them.

\begin{theorem}[No-Cloning Theorem and No-Telegraphing Theorem]\label{thm:no-cloning-no-telegraphing}
Let $\mathcal{H}$ be a Hilbert space, and let $S = \{\ket{\psi_i}\}_{i \in [k]}$ be a collection of pure quantum states on this Hilbert space.
The following are equivalent:
\begin{enumerate}
    \item $S$ can be perfectly cloned
    \item $S$ can be perfectly telegraphed
    \item $S$ is a collection of orthogonal states, with duplication \textnormal{(}$\forall i,j \; \left\lvert \langle{\psi_i}|{\psi_j}\rangle \right\rvert^2$ is either $0$ or $1$\textnormal{)}
\end{enumerate}
\end{theorem}

The proof of cases 1 and 3 is due to~\cite{YUEN1986405} and the addition of case 2 is due to~\cite{Werner_1998}. For completeness, we present the full proof in Appendix~\ref{appendix:no-cloning-equals-no-telegraphing}.
Theorem~\ref{thm:no-cloning-no-telegraphing} demonstrates that a general collection of quantum states cannot be cloned or telegraphed, but all orthogonal collections can. 
% To be specific, this is a special subset, $\Lambda_{\text{Orthogonal}}$, of the set $\Lambda_{\text{All}}$ of all collections of states, and in particular, it contains the set $\Lambda_{\text{Classical}}$ of all collections of classical strings, or equivalently, quantum states in the computational basis ($\Lambda_{\text{Classical}} \subsetneq \Lambda_{\text{Orthogonal}} \subsetneq \Lambda_{\text{All}}$).
% We consider in this paper two subsets of $\Lambda_{\text{Orthogonal}}$, namely 
% $\Lambda^{\text{efficient}}_{\text{Clonable}}$ and 
% $\Lambda^{\text{efficient}}_{\text{Telegraphable}}$, 
% of collections of states which can be respectively cloned or telegraphed \emph{efficiently}, and our main theorem shows that when efficiency is in terms of black-box queries,
% $\Lambda^{\text{efficient}}_{\text{Telegraphable}} \subsetneq \Lambda^{\text{efficient}}_{\text{Clonable}}$.

\subsection{Computational No-go Properties}

We now define the efficient versions of the no-go tasks of cloning and telegraphing, and their associated computational no-go properties.

\paragraph{Computational Restrictions.}
We call the algorithms $\mathsf{Clone}$, $\mathsf{Send}$, and $\mathsf{Receive}$ the adversaries for their respective tasks. Specifying the class of algorithms from which the adversaries may originate allows us to further parameterize the definitions of these no-go tasks by computational complexity.

For instance, if the adversaries are required to be computationally efficient (polynomial-time) quantum algorithms, we say that the scheme is \emph{efficiently} or \emph{computationally} clonable (or unclonable, telegraphable, etc.).
If the scheme includes oracles and the adversaries are quantum oracle algorithms that make a polynomial number of oracle queries, that is, query-efficient algorithms, then we say that the scheme is clonable (unclonable, telegraphable, etc.) by efficient oracle algorithms or query-efficient algorithms.
The one thing to note is that for telegraphing by efficient oracle algorithms, we require as an additional restriction that the classical message $c$ be of polynomial length.
% This not strictly necessary, but simplifies the analysis.
We often use the words ``computational'' and ``efficient'' as a catch-all for both computationally efficient and query-efficient algorithms, and we use more specific terminology whenever it is necessary to differentiate between them.
If the adversaries are not bounded in any way, we say that the scheme is \emph{statistically} or \emph{information-theoretically} clonable  (unclonable, telegraphable, etc.).

\paragraph{Success Probability.}
We say that a scheme is \emph{$\eta$-unclonable} or \emph{$\eta$-untelegraphable} 
% or $\eta$-unreconstructible 
(in either the worst case or in the average case) if no quantum algorithm succeeds at the corresponding task with probability greater than $\eta$.
We will often just drop the parameter $\eta$ and simply say that a scheme is unclonable (or untelegraphable) if it is $\eta$-unclonable (respectively $\eta$-untelegraphable) for every non-negligible probability $\eta$ (non-negligible in the length of the input, in qubits). We say that a scheme is perfectly clonable (or telegraphable) if it is clonable (respectively, telegraphable) with probability 1.

\paragraph{Telegraphing Implies Cloning.}
We now give the trivial direction of the relationship between computational cloning and computational telegraphing: that telegraphing implies cloning. This implication and its proof are certainly not a new result, even in the context of computational efficiency. However, both directions of the relationship have too often been taken for granted despite one direction not always holding. We therefore give a formal proof for the direction that \emph{does} still hold in the context of efficient algorithms, both for completeness, as well as to contrast its simplicity with the relative complexity of the supposed converse.

\begin{theorem}[Telegraphing Implies Cloning]
\label{thm:telegraphing-implies-cloning}
Any scheme that is $(1-\varepsilon)$-computationally telegraphable is also $(1-2\varepsilon)$-computationally clonable. Note that this applies to both computationally efficient and query-efficient algorithms as well as to both worst case and average case versions of these properties.
\end{theorem}
\begin{proof}
We prove this for computationally efficient algorithms, and in the worst case, since the other cases are nearly identical to this one.

Let $S$ be a scheme that is $(1-\varepsilon)$-telegraphable in the worst case by computationally efficient adversaries. That is, there exist efficient quantum algorithms 
$\mathsf{Send}(|\psi\rangle) \rightarrow c$
and 
$\mathsf{Receive}(c) \rightarrow |\phi\rangle$
such that for all $|\psi_z\rangle \in S$, 
$|\phi\rangle := \mathsf{Receive}(\mathsf{Send}(|\psi_z\rangle))$ passes verification for $z$ with probability at least $1-\varepsilon$.

Without loss of generality, we assume that $\mathsf{Send}$ always outputs \emph{some} message to send to $\mathsf{Receive}$. This is because it can always output an arbitrary/random message, which is no worse than outputting nothing.
That is, on input $\ket{\psi_i}$,  $\sum_{c\in\bits^*} \Pr[\mathsf{Send} \text{ outputs } c] = 1$.

The probably of successfully telegraphing is
$$p_{\text{successful telegraphing}} = \sum_c \Pr[\mathsf{Send} \text{ outputs } c] \Pr[\mathsf{Receive} \text{ succeeds on } c] > 1-\varepsilon$$
where the probabilities are taken over the internal randomness and measurements of the algorithms as well as over the randomness of verification.

So, if we run $\mathsf{Send}$ once on $\ket{\psi_i}$ to get message $c$ and then run $\mathsf{Receive}$ twice on the same $c$, we get that the probability of successfully getting two copies of $\ket{\psi_i}$ is
\begin{align*}
    p_{\text{successful cloning}} 
    &= \sum_c \Pr[\mathsf{Send} \text{ outputs } c] \big(\Pr[\mathsf{Receive} \text{ succeeds on } c]\big)^2 \\
    &\ge \left(\sum_c \Pr[\mathsf{Send} \text{ outputs } c] \Pr[\mathsf{Receive} \text{ succeeds on } c]\right)^2 \\
    &\ge (1-\varepsilon)^2 = 1 - 2\varepsilon + \varepsilon^2 > 1 - 2\varepsilon
\end{align*}
Thus $S$ is also $(1-2\varepsilon)$-clonable in the worst case, in time that is at most twice what it took to telegraph.
\end{proof}

% While we incur a loss 
% % (from $\eta$ to $\frac{4}{27} \eta^3$) 
% (from $1 - \varepsilon$ to $1 - 2\varepsilon$)
% in going from telegraphing to cloning, if what we care about is whether $\eta$ is negligible or non-negligible, this loss ends up being insignificant. That is, if we have that a scheme is computationally telegraphable with non-negligible probability, then it is also computationally clonable with non-negligible probability. 
% This is what we mean when we say that telegraphing implies cloning.

Our main result, which we show in Section~\ref{sec:cloning-without-telegraphability}, is that the converse to this theorem does not hold, at least with respect to efficient oracle algorithms.

% % Definition of deconstruction <-- removed for QIP paper

% \begin{definition}[Deconstruction]
% \label{def:deconstruction}
% A scheme $S$ is said to be $\eta$-\textbf{worst case deconstructible} if there exists a quantum algorithm
% $\mathsf{Deconstruct}(|\psi\rangle) \rightarrow c$
% and
% an arbitrary $f : \{0,1\}^* \rightarrow \mathcal Z$
% such that given a quantum state $|\psi_z\rangle \in \mathcal H_z$, 
% $f(\mathsf{Deconstruct}(|\psi_z\rangle)) = z$ with probability at least $\eta$.

% A scheme is said to be $\eta$-\textbf{average case deconstructible} if there exists a quantum algorithm
% $\mathsf{Deconstruct}(|\psi\rangle) \rightarrow c$
% and
% an arbitrary function $f : \{0,1\}^* \rightarrow \mathcal Z$
% such that for
% $(z, |\psi_z\rangle) \leftarrow \mathsf{Gen}()$ 
% $f(\mathsf{Deconstruct}(|\psi_z\rangle)) = z$ with probability at least $\eta$.
% \end{definition}

\subsection{Reconstruction}

Our central aim is to separate efficient cloning from efficient telegraphing. However, in order to do so, we find it convenient to introduce an additional third task, which we call \emph{reconstruction}.

\begin{definition}[Reconstruction]
\label{def:reconstruction}
A scheme $S$ is said to be $\eta$-\textbf{worst case reconstructible} if there exists a quantum algorithm
$\mathsf{Reconstruct}(a) \rightarrow |\phi\rangle$
such that for every label $z \in \mathcal{Z}$, there exists an instance-dependent advice string\footnote{Note that classical tasks become trivial when an adversary is given trusted advice that is \emph{instance-dependent}, as opposed to depending only on the input length. However, the same is not the case for quantum tasks. A quantum task such as that of preparing a quantum state may still be non-trivial, even when given trusted classical advice that depends on each instance.}
$a_z$ such that $|\phi\rangle := \mathsf{Reconstruct}(a_z)$ passes verification for $z$ with probability at least $\eta$.

$(S, \mathcal{D})$ is said to be $\eta$-\textbf{average case reconstructible} if there exists a quantum algorithm
$\mathsf{Reconstruct}(a) \rightarrow |\phi\rangle$
that succeeds at the reconstruction task with probability $\eta$ when $z$ is sampled from the distribution $\mathcal{D}$.

% $(z, |\psi_z\rangle) \leftarrow \mathsf{Gen}()$ and advice $a_z$, $|\phi\rangle := \mathsf{Reconstruct}(a_z)$ passes verification for $z$ with probability at least $\eta$.
\end{definition}

The different parameterized versions of reconstruction are defined analogously to those of cloning and telegraphing.
As with the classical message in the case of telegraphing, for reconstruction by efficient oracle algorithms, we require as an additional restriction that the advice string $a_z$ be of polynomial length.
% In fact, without such a restriction, it could obviate the need to make oracle queries as the truth table of all the oracles could be included in the advice.

Reconstruction can be viewed in one way as a subtask of telegraphing, where we focus our attention only on the receiving end of the telegraphing, or in another way as a telegraphing protocol in which the sender is all-powerful and can implement a (potentially even nonphysical) function from $|\psi_z\rangle$ to $a_z$.
(This function is in fact performing the task of what we call \emph{deconstruction}, which we do not define here, but which can be roughly described as assigning a uniquely identifying label to every state in $S$.)
Following this line of thought, we can observe another trivial implication: between telegraphing and reconstruction.

\begin{theorem}[Telegraphing Implies Reconstruction]
\label{thm:telegraphing-implies-reconstruction}
Any scheme that is $\eta$-computationally telegraphable is also $\eta$-computationally reconstructible. Note that, as before, this applies to both computationally efficient and query-efficient algorithms as well as to both worst case and average case versions of these properties.
\end{theorem}
\begin{proof}
The proof here is even simpler than that of Theorem~\ref{thm:telegraphing-implies-cloning}.
As we did in that proof, we prove this theorem only for computationally efficient algorithms, and in the worst case, since the other cases are much the same.
Let $S$ be a scheme that is $\eta$-telegraphable in the worst case by computationally efficient adversaries. 
That is, there exist efficient quantum algorithms 
$\mathsf{Send}(|\psi\rangle) \rightarrow c$
and 
$\mathsf{Receive}(c) \rightarrow |\phi\rangle$
such that for all $|\psi_z\rangle \in S$, 
$|\phi\rangle := \mathsf{Receive}(\mathsf{Send}(|\psi_z\rangle))$ passes verification for $z$ with probability at least $\eta$.

For every $|\psi_z\rangle \in S$, $\mathsf{Send}(|\psi_z\rangle)$ produces an output $c_z$ that comes from some distribution over classical strings. There must be at least one string $c_z^*$ in its support for which $\mathsf{Receive}(c_z^*)$ succeeds with probability at least $\eta$ (otherwise, $\mathsf{Receive}(c_z)$ has success probability less than $\eta$ for all $c_z$, and so the telegraphing could not have succeeded with probability $\eta$). Thus, for each $z \in \mathcal{Z}$, let $a_z := c_z^*$ and let $\mathsf{Receive}$ be the reconstruction adversary, which we have just shown will succeed on input $a_z$ with probability at least $\eta$ for all $z \in \mathcal{Z}$.
\end{proof}

The direct consequence of Theorem~\ref{thm:telegraphing-implies-reconstruction} is that in order to show that a scheme is not telegraphable, it suffices to show that it is not reconstructible. In other words, in order to prove our separation between computational cloning and computational telegraphing, it suffices to show a scheme that can be computationally cloned but \emph{not computationally reconstructed}.
Reframing our aim in such a way simplifies the analysis because now we only have to deal with a single adversary in both situations (cloning and reconstruction), as opposed to two interacting adversaries for telegraphing. Furthermore, by doing so, we in fact end up showing a stronger separation.

% \begin{theorem}[Informal]
% \label{thm:info_theoretic_equivalence}
% In an information theoretic/statistical setting, Cloning, Telegraphing, and Deconstruction are equivalent to the orthogonality of the set of valid states. On the other hand Reconstruction is trivial.
% \end{theorem}

% \section{Cloning Without Telegraphability}
\section{Cloning without Telegraphability}
\label{sec:cloning-without-telegraphability}

We now come to the main theorem of the paper.

\begin{theorem}
\label{thm:clonable_untelegraphable}
There exists a scheme, relative to a quantum oracle, that on the one hand, can be perfectly cloned by an efficient quantum oracle algorithm in the worst case, but that on the other hand cannot be telegraphed by a pair of efficient quantum oracle algorithms with any non-negligible probability, even in the average case.
\end{theorem}

As mentioned before, we in fact prove the following stronger theorem, which, as a consequence of Theorem~\ref{thm:telegraphing-implies-reconstruction}, implies Theorem~\ref{thm:clonable_untelegraphable}:

\begin{theorem}
\label{thm:clonable_unreconstructible}
There exists a scheme, relative to a quantum oracle, that on the one hand, can be perfectly cloned by an efficient quantum oracle algorithm in the worst case, but that on the other hand cannot be \textbf{reconstructed} by an efficient quantum oracle algorithm with any non-negligible probability, even in the average case.%
\footnote{Note importantly that the fact that these quantum states cannot be efficiently reconstructed does not preclude them from appearing naturally and being used in efficient quantum computation, since they may nevertheless be efficiently \emph{samplable}. That is, there may be an efficient way to sample from the set of states without being able to reconstruct any particular one of them on command. In fact, this is exactly the case for our scheme.}
\end{theorem}

The rest of Section~\ref{sec:cloning-without-telegraphability} contains the proof of Theorem~\ref{thm:clonable_unreconstructible}.
In Section~\ref{sec:the-scheme}, we define the scheme, Scheme~\ref{scheme:clonable_unreconstructible}, and show that it is perfectly clonable.
In Section~\ref{sec:proof-of-unclonability}, we prove that the scheme cannot be efficiently reconstructed.

The form of our scheme is based on a set of states introduced by~\cite{rewinding} which take a uniform superposition over the preimages of a random oracle. These states cannot be cloned by query-efficient algorithms, so by Theorem~\ref{thm:telegraphing-implies-cloning} this directly implies that they are untelegraphable.\footnote{Note, however, that this does not imply that they are unreconstructable. Nevertheless, we show that this is the case in Proposition~\ref{prop:Stage1}.}
We want a scheme that is untelegraphable despite being clonable, so we add a cloning oracle, a quantum oracle that clones only this set of states. The main technical challenge is to show that access to this cloning oracle does not allow the adversaries to telegraph.

We start by showing that with just the random oracle, the states are not reconstructible, via a reduction from the problem of finding multi-collisions in the random oracle. We then show that allowing cloning for the target state cannot be detected by the adversary. We finally simulate the rest of the cloning oracle by replacing the random oracle with an impostor for which we know how to clone.

\subsection{The Scheme}
\label{sec:the-scheme}

Before we give the scheme, we first give a few definitions that are useful both for defining the scheme and for the proof of its unreconstructibility.

We first define a cloning oracle for orthonormal sets. This is an oracle that successfully clones a specific subset of basis states for a given basis.

\begin{definition}[Cloning oracle for a set] \label{def:cloning_oracle}
Let $\mathcal{H}$ be a Hilbert space and let $S = \{|\psi_i\rangle\}_{i\in[k]}$ be an orthonormal subset of $\mathcal{H}$. Augment $\mathcal{H}$ with a special symbol $\bot$ outside the support of $\mathcal H$. That is, $|\bot\rangle$ is orthogonal to all of $\mathcal{H}$.

A \textbf{cloning oracle} $\mathcal{C}_{S}$ on set $S = \{|\psi_i\rangle\}_{i\in[k]}$ is a quantum oracle that, for all $i \le k$ 
sends $|\psi_i\rangle |\bot\rangle$ to $|\psi_i\rangle |\psi_i \rangle$ and $|\psi_i\rangle |\psi_i \rangle$ to $|\psi_i\rangle |\bot \rangle$. For all other orthogonal states, it applies the identity.
That is, when the second register is $|\bot\rangle$, it clones any state in $S$ and leaves all other orthogonal states unmodified.
\end{definition}

\begin{definition}[Preimage superposition state]
\label{def:preimage-superposition-state}
Let $f:\{0,1\}^m \rightarrow \{0,1\}^n$. A \textbf{preimage superposition state} for image $z \in \{0,1\}^n$ in function $f$ is the quantum state that is the uniform positive superposition of preimages of $z$ in $f$:
$$|\psi_{z}\rangle = \frac{1}{\sqrt{|f^{-1}(z)|}} \sum_{x | f(x) = z} |x\rangle$$
where $f^{-1}(z) := \{x | f(x) = z\}$ is the set of preimages of $z$ in $f$.
\end{definition}

\begin{definition}[Preimage superposition set] \label{def:preimage-superposition-set}
Let $f:\{0,1\}^m \rightarrow \{0,1\}^n$. A \textbf{preimage superposition set for} $f$, $S_f$, is the set of preimage superposition states for all images in the range of~$f$.
$$S_f := \left\{ \frac{1}{\sqrt{|f^{-1}(z)|}} \sum_{x | f(x) = z} |x\rangle \; \Bigg| \; z \in \{0,1\}^n\right\}$$
\end{definition}

\begin{definition}[Cloning oracle relative to a function]
Let $f:\{0,1\}^m \rightarrow \{0,1\}^n$. A \textbf{cloning oracle relative to $f$}, 
$\mathcal{C}_{f}$,
is a cloning oracle for the preimage superposition set, $S_f$, of $f$.
\end{definition}

We now give the formal definition of the scheme:

\begin{scheme} \label{scheme:clonable_unreconstructible}
Let $H : \{0,1\}^m \rightarrow \{0,1\}^n$ be a random oracle, where $m \ge 2n$ (but bounded by a polynomial in $n$).
Let $\mathcal{C}_H$ be the cloning oracle relative to $H$. The scheme consists of the following:

\begin{itemize}
    \item[--] The collection of oracles is $\mathcal{O} := \{H, \mathcal{C}_H\}$.

    \item[--] The set of states is $S := S_H$, the preimage superposition set for $H$.

    \item[--] The distribution, $\mathcal{D}$, samples the image of a random domain element of $H$. That is, it returns $z \leftarrow H(x)$ for a uniformly random $x \in \{0,1\}^m$.
\end{itemize}
\end{scheme}

It is clear that the scheme presented here is perfectly clonable in the worst case by an efficient quantum oracle algorithm. Specifically, the cloning oracle, $\mathcal{C}_H$, provides that capability, and in a single oracle query. Therefore, it remains to show that no efficient quantum oracle algorithm can reconstruct it. This is the main technical challenge of our proof and takes up the remaining part of Section~\ref{sec:cloning-without-telegraphability}.%
\footnote{
As is evident from Scheme~\ref{scheme:clonable_unreconstructible}, we prove Theorems~\ref{thm:clonable_untelegraphable} and~\ref{thm:clonable_unreconstructible} relative to a quantum oracle (or rather, a pair of quantum oracles) sampled from a probability distribution rather than a fixed quantum oracle. However, as mentioned in the introduction to the paper, this is not a weakness, as a straightforward transformation similar to the proof of 
Theorem~\ref{thm:clonableqma-qcma-oracle-separation} 
allows fixing the randomness at the cost of the proof becoming non-constructive. Moreover, it is not necessary to show this explicitly outside of 
Theorem~\ref{thm:clonableqma-qcma-oracle-separation}, as it is already implied by the combination of
Theorems~\ref{thm:clonableqma-qcma-oracle-separation}~and~\ref{thm:complexity-separation-to-nogo-separation}.
}

\subsection{Proof of Unreconstructibility}
\label{sec:proof-of-unclonability}

We wish to prove that Scheme~\ref{scheme:clonable_unreconstructible} cannot be efficiently reconstructed by efficient quantum oracle algorithms in the average case.
We prove this in a sequence of three stages, beginning with a simplified version of the scheme without a cloning oracle, then moving to one with an oracle that can only clone a single state, and finally to the full scheme with the full cloning oracle.
% then building up step by step

\subsubsection{With No Cloning Oracle}

In the first stage, we consider an adversary, $R$, which is a quantum oracle algorithm with advice. $R$ is given a polynomial length advice string $a_z$, and is allowed a polynomial number of queries to the random oracle. It is tasked with producing a state that passes verification for $z$, namely the positive uniform superposition over all the preimages of $z$ in the random oracle.
Note that this first version does not yet have access to a cloning oracle of any sort.

\begin{proposition} \label{prop:Stage1}
Let $R$ be a quantum oracle algorithm that is given a classical advice string $a_z \in \{0,1\}^\ell$ for some polynomial $\ell$ in $n$, and makes $q$ queries to the random oracle, where $q$ is a polynomial in $n$. For $z \in \{0,1\}^n$ drawn uniformly at random, $R$ cannot output a quantum state that passes verification for $z$ with probability that is non-negligible in $n$.%
\footnote{
Note that the advice string, $a_z$, may in general contain \emph{any} information, including, for instance, any details about the set of preimages of $z$ in $H$, or any other useful information about the task. We show here that no polynomial amount of classical information \emph{of any kind} will allow $R$ to faithfully reconstruct the state.
}
\end{proposition}

\begin{proof}
The main idea is that if $R$ were able to produce the target state $|\psi_z\rangle$ with non-negligible probability, then it can also do so without the advice by guessing the advice string, albeit with significantly lower probability. Measuring $|\psi_z\rangle$ then gives a random preimage of the random oracle, and we can do this multiple times to produce several preimages of the same image $z$, producing a multi-collision for the random oracle, which is harder to do than this method would give.

We now give the proof. Suppose, for the sake of contradiction, that $R$, when given advice string $a_z$, makes $q$ queries to the random oracle and then produces the mixed state $\rho_z$ which passes verification for $z$ with non-negligible probability $\eta$ (that is, $\bra{\psi_z} \rho_z \ket{\psi_z} \ge \eta$). We use $R$ to produce a large number of disjoint collisions of the oracle.

% Let $k$ be a sufficiently large polynomial in $n$, for instance let $k=2n(\ell+1)$ (note that $\ell$ is itself a positive integer bounded by a polynomial in $n$). We run $R$ repeatedly (on the same target label $z$ and advice $a_z$) a total of $4k/\eta$ times with the goal of producing at least $2k$ valid copies of $|\psi_z\rangle$. 
% % $$P[X \le (1 - \delta)\mu] \le e^{-\mu\delta^2/2} \quad \text{ for all } 0 < \delta < 1$$ % Chernoff bound used here
% By a Chernoff bound, this succeeds with constant probability $\Omega(1)$: that is, if $X$ is the number of valid instances, $P\left[X \le \frac12(4k) \right] \le e^{-4k/8} = e^{-n(\ell+1)} \le 1/2$. 

% Since the only accepting state of the scheme is the positive uniform superposition over the preimages of $z$, measuring each resulting copy gives an independent uniformly random preimage of $z$. 
% Measuring all $2k$ copies then gives $2k$ independent samples from this set of preimages. We want all the samples to be unique, with no repeats, and this happens with high probability at least $\Omega(1)$.
% (That is, with high probability, $z$ has $\Omega(2^{m-n})$ preimages in the random oracle. Since the samples are drawn independently, the probability that any pair of samples coincides is $O(2^{-(m-n)})$. By a union bound over all polynomially many pairs of samples, the probability that any pair coincides is still $2^{-\Omega(m-n)}$, which is negligible, since we have $m \ge 2n$. This means that all samples are unique with high probability.)

Let $H^{-1}(z)$ be the set of preimages of $z$ in $H$. We have that with high probability, $|H^{-1}(z)| \ge \Omega(2^{m-n})$.
Let $\Gamma \subset H^{-1}(z)$ be an arbitrary polynomial sized subset of $H^{-1}(z)$, and let $\Pi$ be the binary projective measurement that projects onto the preimages of $z$ that are not in $\Gamma$, that is, onto the computational basis states $H^{-1}(z) \setminus \Gamma$.
We have that $\bra{\psi_z} \Pi \ket{\psi_z} \ge 1 - \epsilon$ for $\epsilon = \frac{|\Gamma|}{|H^{-1}(z)|} \in \negl(n)$. 
Given that $\bra{\psi_z} \rho_z \ket{\psi_z} \ge \eta$, we apply Lemma~\ref{lemma:composition} to get that $\tr(\Pi \rho_i) \ge \eta (1-\epsilon) - 2 \sqrt{\epsilon(1-\eta)} \ge \frac12 \eta$ for sufficiently large $n$. In other words, for any polynomial sized subset of preimages, and for sufficiently large $n$, we have that measuring $\rho_z$ will with non-negligible probability give a preimage of $z$ outside that subset.

Let $k$ be a sufficiently large polynomial in $n$, for instance let $k=2n(\ell+1)$ (note that $\ell$ is itself a positive integer bounded by a polynomial in $n$). We run $R$ repeatedly (on the same target label $z$ and advice $a_z$) a total of $8k/\eta$ times and measure the outcome in the computational basis, with the goal of producing at least $2k$ unique preimages of $z$. 
% $$P[X \le (1 - \delta)\mu] \le e^{-\mu\delta^2/2} \quad \text{ for all } 0 < \delta < 1$$ % Chernoff bound used here
By a Chernoff bound, this then succeeds with constant probability $\Omega(1)$: that is, if $X$ is the number of valid unique preimages, $\Pr\left[X \le \frac12(4k) \right] \le e^{-4k/8} = e^{-n(\ell+1)} \le 1/2$. 
Finally, because every pair of unique preimages is a collision, this gives $k$ disjoint collisions of the random oracle.
That is, this process therefore produces $k$ disjoint collisions with constant probability $\Omega(1)$.

Now, if this process succeeds given the advice $a_z \in \{0,1\}^\ell$, then it can also succeed without being given advice, though with a much lower probability, by guessing the advice string with probability $2^{-\ell}$, for an overall success probability of at least $\Omega(2^{-\ell})$. 

To recap, this gives an quantum oracle algorithm for producing $k$ disjoint collisions of a random oracle which makes $t = 8k q/\eta$ oracle queries and succeeds with probability at least $\Omega(2^{-\ell})$.

On the other hand, we recall the following theorem from Hamoudi and Magniez~\cite{Hamoudi2020-mz}:
\begin{theorem}[Theorem 4.6 from~\cite{Hamoudi2020-mz}] \label{thm:HM20_4.6}
The success probability of finding $K$ disjoint collisions in a random function
$f : [M] \rightarrow [N]$ is at most $O(T^3/(K^2N))^{K/2} + 2^{-K}$ for any algorithm making $T$ quantum queries to $f$ and any $1 \le K \le N/8$.
\end{theorem}

Applying the bound from the above Theorem~\ref{thm:HM20_4.6} with $T=8k q/\eta$, $K=k$, $M=2^m$ and $N=2^n$, the success probability for this task must be at most
\begin{align*}
    O\left(\frac{T^3}{K^2N}\right)^{K/2} + 2^{-K} 
    &= O\left(\frac{(8k q/\eta)^3}{k^2 2^n}\right)^{k/2} + 2^{-k} \\
    &= O\left(\frac{k q^3}{\eta^3 2^n}\right)^{k/2} + 2^{-k} \\
    &\le 2^{-\Omega\left(k\right)} \\
    &\le 2^{-\Omega\left(n(\ell+1)\right)}
\end{align*}

There therefore exists a sufficiently large $n$, for which this is a contradiction. This completes the proof of Proposition~\ref{prop:Stage1}.
\end{proof}

\subsubsection{With a Limited Cloning Oracle}
In the second stage, we allow R access to a limited cloning oracle which can clone only the target state.

% % old definition
% \begin{definition}
% Let $z$ be a key and $|\psi_z\rangle \in \mathcal H_z$. Let $\bot$ be a symbol outside the support of $\mathcal H$. A \textbf{$z$-cloning oracle} is a quantum oracle which sends $|\psi_z\rangle |\bot\rangle$ to $|\psi_z\rangle |\psi_z\rangle$ and vice versa, and acts as the identity on all other orthogonal states. That is, it only clones the state $|\psi_z\rangle$ and leaves all other orthogonal states unmodified.
% \end{definition}

\begin{definition}\label{def:z-cloning-oracle}
Let $z$ be a label and $|\psi_z\rangle$ the corresponding quantum state from the scheme. A~\textbf{$z$-cloning oracle}, $\mathcal{C}_{z}$, is a cloning oracle for the singleton set $\{|\psi_z\rangle\}$.
\end{definition}

\begin{proposition}\label{prop:adding-the-z-cloning-oracle}
Let $R$ be a quantum oracle algorithm that is given a classical advice string $a_z \in \{0,1\}^\ell$ for some polynomial $\ell$ in $n$, and makes $q$ queries (where $q$ is a polynomial in $n$) to the random oracle \textbf{as well as a $z$-cloning oracle}.
Let $R'$ be a run of $R$ where queries to the $z$-cloning oracle are instead returned unmodified (or equivalently, passed to a dummy oracle which acts as the identity).
Then the total variation distance between the outcomes of the two runs is negligible in $n$.
\end{proposition}
\begin{proof}
The idea is that since the $z$-cloning oracle and the dummy oracle differ only on the basis states where the first register is $|\psi_z\rangle$,
% $|\psi_z\rangle |\bot\rangle$ and $|\psi_z\rangle |\psi_z\rangle$,
if $R$ puts low query weight on those basis states, then swapping between the oracles can only make minimal difference.

We now give the proof. Consider an adversary $R$ which, when given advice string $a_z$, makes 
% $q_{r}$ 
% $q$
$k$
queries to the random oracle and 
$q$ 
queries to the $z$-cloning oracle. 
% and then produces a state $\rho$ which 
% passes some binary projective measurement $\Pi$
% % passes verification for $z$ 
% with probability $\eta$. That is, $\tr(\Pi\rho) \ge \eta$.
% We use $R$ to construct a similar algorithm, $R'$, which does not make use of the $z$-cloning oracle, and we show that it cannot behave too differently from the one which does.
Let $R'$ be a quantum oracle algorithm which simulates a run of $R$ in which the $z$-cloning oracle is replaced by a dummy oracle (an oracle which acts as the identity on all states) by ignoring all of the $z$-cloning oracle queries (equivalent to performing the identity on each one).

For each $t \in [q]$, let $R'_t$ be a version of $R'$
% similarly a simulation of a run of $R$ using the dummy cloning oracle, but 
in which the simulation is stopped prematurely at cloning query number $t$ and which then outputs the first register of the query input. 
% Let $|\phi_t\rangle$ be the outputted state and let $\eta'_t$ be the corresponding success probability of $R'_t$. As above, since each $R'_t$ is a quantum oracle algorithm with advice that satisfies the conditions of Proposition~\ref{prop:Stage1}, all the 
% $\{\eta'_t\}_{t\in[q]}$ 
% must be negligible. 
% More specifically, 
% $\eta'_t = |\langle\psi_z|\phi_t\rangle|^2$,
% so the overlap between $|\phi_t\rangle$ and $|\psi_z\rangle$ must be negligible. 
Let $\rho'_t$ be the reduced density matrix of the outputted state.
Because the runs of $R'$ and $R'_t$ are identical up until $R'_t$ stops and outputs cloning query number $t$, $\rho'_t$ is also the reduced state of the first query register when $R'$ requests query number $t$.
% The probability that $R'_t$ succeeds in verification is
Let $\eta'_t := \langle\psi_z|\rho'_t|\psi_z\rangle$ be the probability that $\rho'_t$ would pass verification for $z$.
As above, since each $R'_t$ is a quantum oracle algorithm with advice that satisfies the conditions of Proposition~\ref{prop:Stage1}, all the 
$\eta'_t$ must be negligible in $n$.

% % Previous version for the insecure cloning oracle that's equivalent to XOR

% Let $\{|\psi_i\rangle\}_{i\in [0,\dim{(\mathcal{H})}]}$ be a basis for $\mathcal{H}' := \mathcal{H}\oplus|\bot\rangle$ as in Definition~\ref{def:cloning_oracle}, whith $|\psi_0\rangle := |\bot\rangle$, and $|\psi_1\rangle := |\psi_z\rangle$. In this basis, both the $z$-cloning oracle and the dummy oracle are applications of classical functions. Specifically, the $z$-cloning oracle becomes an application of the classical indicator function for $1$: $f_{=1}(x) = \mathbb I[x = 1]$, and the dummy cloning oracle becomes an application of the all-zero function $f_{\emptyset}(x) = 0$.

Choose a basis $\{|\chi_i\rangle\}_{i\in [0,\dim{(\mathcal{H})}]}$ for $\mathcal{H}' := \mathcal{H}\oplus|\bot\rangle$ as in Definition~\ref{def:cloning_oracle}, that includes $|\chi_0\rangle := |\bot\rangle$ and $|\chi_z\rangle := |\psi_z\rangle$ as two basis elements. 
Let $D$ be a unitary from this basis into the computational basis that sends $|\psi_z\rangle|\bot\rangle$ to $|z\rangle|0\rangle$ and $|\psi_z\rangle|\psi_z\rangle$ to $|z\rangle|1\rangle$ and which arbitrarily assigns all other orthogonal states to computational basis states.
(For example, let $B = \sum_{i} |i\rangle\langle \chi_i|$, let $C = \sum_{ij \notin \{(z,1), (z,z)\}} |i\rangle|j\rangle\langle i|\langle j| + |z\rangle|1\rangle\langle z|\langle z| + |z\rangle|z\rangle\langle z|\langle 1|$, and let $D = C \cdot (B\otimes B)$.)

In this basis, both the $z$-cloning oracle and the dummy oracle can be expressed as applications of binary classical functions on all but the last bit. The $z$-cloning oracle becomes an application of the classical indicator function for the string $(z,0^{m-1})$: $f_{=z}(x) = \begin{cases}1 & x = (z,0^{m-1}) \\ 0 & \text{otherwise}\end{cases}$, and the dummy cloning oracle becomes an application of the all-zero function $f_{\emptyset}(x) = 0$.
Let $\mathcal{O}_{=z}$ be the unitary application of the indicator function, $f_{=z}(x)$ above, which XORs the result into the last bit. Then the $z$-cloning oracle can be expressed as $\mathcal{C}_{z} =  D^\dagger \mathcal O_{=z} D$, and the dummy oracle can be expressed as $\mathcal{C}_{\emptyset} =  D^\dagger \mathcal{O}_{\text{identity}} D = D^\dagger I D = I$.

The algorithms $R$, $R'$, and $\{R'_t\}_{t\in[q]}$ can therefore be reformulated as quantum oracle algorithms that direct cloning queries to the classical oracles $\mathcal{O}_{=z}$ in the case of $R$, or $\mathcal{O}_{\text{identity}}$ in the case of $R'$ and $\{R'_t\}_{t\in[q]}$. That is, before they make a cloning query, they apply the change of basis $D$ into the computational basis. They then query $\mathcal{O}_{=z}$ or $\mathcal{O}_{\text{identity}}$, and then apply the change of basis $D^\dagger$ back to the original basis. Call the versions of $R$, $R'$, and $\{R'_t\}_{t\in[q]}$ in this new basis $\mathcal R$, $\mathcal R'$, and $\{\mathcal R'_t\}_{t\in[q]}$. 

Note that the only difference between $\mathcal R$ and $\mathcal R'$ is that cloning queries to $\mathcal{O}_{=z}$ in $\mathcal R$ are redirected to $\mathcal{O}_{\text{identity}}$ in $\mathcal R'$. Furthermore, $\mathcal{O}_{=z}$ and $\mathcal{O}_{\text{identity}}$ only differ on inputs where the first register is $|z\rangle$.

We now therefore compute the query magnitude of cloning queries of $\mathcal R'$ on $|z\rangle$. The state of the first register of cloning query number $t$ of $\mathcal R'$ to $\mathcal{O}_{\text{identity}}$ is given by 
$B \rho'_t B^\dagger =  \sum_{i,j} |i\rangle\langle \chi_i| \; \rho'_t \; |\chi_j\rangle\langle j|$.
The query magnitude on $|z\rangle$ is then $\langle z| \left( \sum_{i,j} |i\rangle\langle \chi_i| \; \rho'_t \; |\chi_j\rangle\langle j| \right) |z\rangle = \langle \chi_z| \rho'_t |\chi_z\rangle  = \langle \psi_z| \rho'_t |\psi_z\rangle = \eta'_t$.

We now apply Theorem~\ref{thm:Ben+97_3.3} on the set $F := \{(i,y) \; | \; i\in [q], y = z\}$ and $\varepsilon := \sqrt{T \sum_{t=1}^{T} \eta'_t}$, with $T := q$.
The sum of the query magnitudes of $\mathcal R'$ on $F$ is then $\sum_{t=1}^{T} \eta'_t \le \frac{\varepsilon^2}{T}$. 
Let $|\phi\rangle$ and $|\phi'\rangle$ be the states outputted by $\mathcal R$ and $\mathcal R'$ respectively (and therefore also by $R$ and $R'$ respectively).
Since $\mathcal R$ is identical to $\mathcal R'$, with the only difference being that the cloning oracle queries are modified on the set $F$, then by Theorem~\ref{thm:Ben+97_3.3}, $\big| |\phi\rangle - |\phi'\rangle \big| \le \varepsilon$. 
By Theorem~\ref{thm:Ben+97_3.1}, then, the total variation distance between runs of $R$ and $R'$ is therefore at most $4\varepsilon = 4\sqrt{q \sum_{t=1}^{q} \eta'_t}$, which is negligible, since all the $\eta'_t$ are negligible and $q$ is a polynomial.
\end{proof}

\begin{corollary} \label{cor:Stage2}
Let $R$ be a quantum oracle algorithm that is given a classical advice string $a_z \in \{0,1\}^\ell$ for some polynomial $\ell$ in $n$, and makes $q$ queries (where $q$ is a polynomial in $n$) to the random oracle \textbf{as well as a $z$-cloning oracle}. For $z \in \{0,1\}^n$ drawn uniformly at random, $R$ cannot output a quantum state that passes verification for $z$ with probability that is non-negligible in $n$.
\end{corollary}

\begin{proof}
Suppose that $R$, when given advice string $a_z$, makes $k$ queries to the random oracle and $q$ 
queries to the $z$-cloning oracle, and then produces a state $\rho_z$ which passes verification for $z$ with probability $\eta$.
As in Proposition~\ref{prop:adding-the-z-cloning-oracle}, let $R'$ be a run of $R$ in which queries to the $z$-cloning oracle are returned unmodified.
$R'$ is then a quantum oracle algorithm with advice that satisfies the conditions of Proposition~\ref{prop:Stage1}, so it must have negligible success probability $\eta'$.
By Proposition~\ref{prop:adding-the-z-cloning-oracle}, the total variation distance between runs of $R$ and $R'$ is negligible in $n$ so $\eta$ is at most negligibly larger than $\eta'$, and thus negligible as well.
% This completes the proof of Corollary~\ref{cor:Stage2}.
\end{proof}

\subsubsection{With the Full Cloning Oracle}
In the third stage, we finally allow $R$ access to the full cloning oracle, which clones all valid states of the scheme while doing nothing for invalid states.

\begin{proposition} \label{prop:Stage3}
Let $R$ be a quantum oracle algorithm that is given a classical advice string $a_z \in \{0,1\}^\ell$ for some polynomial $\ell$ in $n$, and makes $q$ queries (where $q$ is a polynomial in $n$) to the random oracle \textbf{and a full cloning oracle} for the set of valid states. For $z \in \{0,1\}^n$ drawn uniformly at random, $R$ cannot output a quantum state that passes verification for $z$ with probability that is non-negligible in $n$.
\end{proposition}

\begin{proof}
Note that in showing this, we are demonstrating that the ability to clone other valid states does not help it produce the target state. 
The idea is use $R$ to produce a new adversary $R'$ which queries just the $z$-cloning oracle with comparable success.
Ideally we would take a $z$-cloning oracle and simply simulate the rest of the cloning oracle (for states other than the target state) by using the random oracle. 
However, such a simulation would require a large number of queries to the random oracle and thus be highly inefficient. 
We get around this issue by creating an imposter random oracle and simulating cloning queries relative to it rather than relative to the original random oracle. We must show first that the impostor random oracle is indistinguishable from the original random oracle, and second that it is possible to approximately simulate cloning queries to the impostor oracle.

We now give the proof. Consider an adversary $R$ which, when given advice string $a_z$, makes 
$q$
queries to the random oracle and the full cloning oracle, and then with probability $\eta$ produces a state $|\psi_z\rangle$ which passes verification for $z$.
We use $R$ to produce a similar algorithm, $R'$, which only makes cloning queries to the $z$-cloning oracle, and which must succeed with comparable probability. We produce $R'$ as follows:

We first sample a private random function, $H_{\text{private}}: \{0,1\}^m \rightarrow \{0,1\}^n \setminus \{z\}$, which has a limited codomain such that it does not output $z$.
That is, for each input, independently choose a uniformly random element of $\{0,1\}^n \setminus \{z\}$.

We then create an impostor random oracle, $H_{\text{impostor}}: \{0,1\}^m \rightarrow \{0,1\}^n$, by combining the original and private random oracles in the following way:
\begin{align*}
    H_{\text{impostor}}(x) = \begin{cases}
    z & H(x) = z \\
    H_{\text{private}}(x) & \text{otherwise}
    \end{cases}
\end{align*}

That is, on query input $x$, if $x$ is a preimage of $z$ in $H$, it passes the query to the original random oracle, producing $z$, but otherwise passes it to the newly sampled private random oracle.

We also create a cloning oracle relative to this impostor random oracle, $\mathcal C_{\text{impostor}}$.
This \textit{impostor cloning oracle} clones the states that are valid for the impostor random oracle, which will in general be different than the set of valid states of the original random oracle. We claim that the impostor oracles perfectly mimic the originals.

\begin{claim}
    The joint distribution of target image $z$ and the impostor random oracle
    $H_{\text{impostor}}$ is identical to that of $z$ and $H$.
    That is, $H_{\text{impostor}}$ is distributed as a uniform random oracle conditioned on $z$ being one of its images.
\end{claim}
\begin{proof}
To show this, we begin giving an equivalent lazy method of sampling the random oracle $H$, along with sampling the target image $z$.

First, we choose a random element $x^* \in \{0,1\}^m$ in the domain of $H$.
We then randomly choose $z \in \{0,1\}^n$ as both its image in $H$ and as the target image. Then, for each of the remaining elements of the domain of $H$, sample a random image from its range.

We now describe a similar method for lazily sampling the impostor random oracle $H_{\text{impostor}}$, along with the target image $z$.

As before, we choose a random element $x^* \in \{0,1\}^m$ in the domain, and a random image $z \in \{0,1\}^n$ as both its image in $H_{\text{impostor}}$ and as the target image. For each remaining element, $x$, of the domain, we first sample a random image $y$. If $y \ne z$, resample an independent sample $y'$ from $\{0,1\}^n \setminus \{z\}$ to be the image of $x$.
Since, conditioned on $y \ne z$, $y$ is uniform on $\{0,1\}^n \setminus \{z\}$, and so is $y'$, the resampled image $y'$ is identically distributed to the original $y$. The extra resampling performed to sample $H_{\text{impostor}}$ thus has no effect on the distribution, so this process produces a distribution identical to the one above for sampling $H$ and $z$.
\end{proof}

As a consequence, no quantum oracle algorithm can tell the difference between query access to the original oracles $H$ and $\mathcal C_H$, and query access to the impostor oracles $H_{\text{impostor}}$ and $\mathcal C_{\text{impostor}}$. That is, an algorithm $R''$ which simulates $R$ and redirects its oracle queries to the impostor oracles will succeed with the same probability $\eta$.

This completes the first part, showing that the impostor oracles are perfect replacements for the original oracles. It now remains to show that the impostor oracles can be simulated efficiently in terms of the number of queries to the original random oracle $H$ and a $z$-cloning oracle $\mathcal{C}_z$.

Note that implementing $\mathcal{C}_{\text{impostor}}$ using $H$ and $\mathcal C_z$ may be query inefficient. We therefore create a new efficient impostor cloning oracle $\widehat{\mathcal{C}}_{\text{impostor}}$, which for each query only makes a constant number of queries to $H$ and $\mathcal C_z$, but which nevertheless performs nearly as well as the inefficient $\mathcal{C}_{\text{impostor}}$.

We would like to define $\widehat{\mathcal{C}}_{\text{impostor}}$ by saying that it acts on computational basis states approximately as
$$\widehat{\mathcal{C}}_{\text{impostor}} |x\rangle |y\rangle = \begin{cases}
\mathcal{C}_{z} |x\rangle |y\rangle & H(x) = z \\
\mathcal{C}_{\text{private}} |x\rangle |y\rangle & \text{otherwise}
\end{cases}$$
However, in reality, this is not unitary, since the resulting states will not be exactly orthogonal. 
We instead define it with an additional ancilla qubit as follows:

Define the following two unitaries acting on an ancilla qubit $|b\rangle$ as well as the two input registers (of the cloning oracle).

\begin{align*}
\mathcal{U}_1 |b\rangle |x \rangle |y\rangle &= \begin{cases}
|b \oplus 1\rangle |x\rangle |y\rangle & H(x) = z \\
|b\rangle |x\rangle |y\rangle & \text{otherwise}
\end{cases}\\
\mathcal{U}_2 |b\rangle |x \rangle |y\rangle &= \begin{cases}
|b\rangle \otimes \mathcal{C}_{z} |x\rangle |y\rangle & b = 1 \\
|b\rangle \otimes \mathcal{C}_{\text{impostor}} |x\rangle |y\rangle & \text{otherwise}
\end{cases}
\end{align*}

The action of $\mathcal{C}_{\text{impostor}}$ with an extra ancilla qubit can be expressed as 
$I \otimes \mathcal{C}_{\text{impostor}} \ket{0}\ket{x}\ket{y} = \mathcal{U}_1\mathcal{U}_2\mathcal{U}_1 \ket{0}\ket{x}\ket{y}$. 
That is, after applying $\mathcal{U}_1$, then $\mathcal{U}_2$ acts as $I \otimes \mathcal{C}_{\text{impostor}}$ because whenever $H(x) = z$, then $\mathcal{C}_{z} |x\rangle |y\rangle = \mathcal{C}_{H} |x\rangle |y\rangle = \mathcal{C}_{\text{impostor}} |x\rangle |y\rangle$. Furthermore, for any $x \in \{0,1\}^m$, the support of the state $\mathcal{C}_{\text{impostor}} |x\rangle |y\rangle$ is only on computational basis states $|x'\rangle |y'\rangle$ such that $H(x') = z \Leftrightarrow H(x) = z$, which implies that the second application of 
$\mathcal{U}_1$ properly uncomputes its own action on the ancilla qubit.
% $\mathcal{U}_2 = I \otimes \mathcal{C}_{\text{impostor}}$ commutes with $\mathcal{U}_1$.

We now define a modified version of $\mathcal{U}_2$, but which makes no use of $\mathcal{C}_{\text{impostor}}$, and instead uses $\mathcal{C}_{\text{private}}$, the cloning oracle relative to $H_{\text{private}}$:

\begin{align*}
\widehat{\mathcal{U}}_2 |b\rangle |x \rangle |y\rangle &= \begin{cases}
|b\rangle \otimes \mathcal{C}_{z} |x\rangle |y\rangle & b = 1 \\
|b\rangle \otimes \mathcal{C}_{\text{private}} |x\rangle |y\rangle & \text{otherwise}
\end{cases}
\end{align*}

We thus define $\widehat{\mathcal{C}}_{\text{impostor}} = \mathcal{U}_1\widehat{\mathcal{U}}_2\mathcal{U}_1$, 
which we note makes two queries to $H$ and one query\footnotemark{} to $\mathcal C_z$ on each application
(note that $\mathcal{C}_{\text{private}}$ uses no oracle queries as it can be simulated directly using the private random function $H_{\text{private}}$). 
It remains to show that, when the ancilla qubit is initialized to $\ket{0}$, $\widehat{\mathcal{C}}_{\text{impostor}}$ cannot be distinguished from $I \otimes \mathcal{C}_{\text{impostor}}$.
That is, we show that it is a good efficient approximation for $\mathcal{C}_{\text{impostor}}$.
\footnotetext{Note that it is straightforward to implement a controlled version of $\mathcal{C}_z$ using a single query to $\mathcal{C}_z$, as it is for any oracle for which a fixed state, on which it acts as the identity, is known. In this case, the fixed state is $|\bot\rangle|\bot\rangle$. To do so, we prepare the state $|\bot\rangle|\bot\rangle$ in an ancilla register. We then apply a 0-controlled SWAP gate between this register and the input register on which $\mathcal{C}_z$ acts, once before and then then once again after applying $\mathcal{C}_z$. If the control is a 0, then the fixed state $|\bot\rangle|\bot\rangle$ is swapped in, neutralizing the application of the oracle. If the control is a 1, then nothing is swapped and the oracle acts as expected.}

We observe that the actions of $I \otimes \mathcal{C}_{\text{impostor}}$ and $\widehat{\mathcal{C}}_{\text{impostor}}$ 
differ only in whether they apply $\mathcal{C}_{\text{impostor}}$ or $\mathcal{C}_{\text{private}}$ 
(in $\mathcal U_2$ and $\widehat{\mathcal U}_2$ respectively) 
on the two non-ancilla registers, 
and only on basis states for which the first of those registers is not a preimage of $z$.
In fact they differ only by a change of basis between a basis that includes the preimage superposition set of $H_{\text{private}}$ and one that includes the preimage superposition set of $H_{\text{impostor}}$.

Taking a closer look at $H_{\text{impostor}}$ and $H_{\text{private}}$, the only difference between the functions is that the domain elements that are preimages of the target image $z$ in $H_{\text{impostor}}$ are reassigned to another random image in $H_{\text{private}}$. Moreover, since the difference we observe in this setting is only for domain elements that do not map to $z$ in $H$ (and thus in $H_{\text{impostor}}$), we can set aside $z$ in the analysis and focus on the other images.

Let $|\psi_i\rangle$ and $|\widehat\psi_i\rangle$ be the respective preimage superposition states of $H_{\text{impostor}}$ and $H_{\text{private}}$ for image $i \in \{0,1\}^n\setminus\{z\}$. 
Let $\theta_i := \cos^{-1}\left(\langle \psi_i | \widehat\psi_i \rangle\right)$ be the small angle between them.
Further, let $|\psi_{z\rightarrow i}\rangle$ be the equal positive superposition over any preimages of the target image, $z$, in $H_{\text{impostor}}$ that were reassigned to image $i$ in $H_{\text{private}}$. 
Then we can write $|\widehat\psi_i\rangle =  \cos (\theta_i) |\psi_i\rangle + \sin (\theta_i) |\psi_{z\rightarrow i}\rangle$.

Note that for all $i \ne j$, $\langle\psi_i|\widehat\psi_j\rangle = \cos (\theta_j) \langle\psi_i|\psi_j\rangle + \sin (\theta_j) \langle\psi_i|\psi_{z\rightarrow j}\rangle = 0$ because the supports of the states (that is, their sets of preimages) are disjoint (where note again that we exclude the target image $z$ here). And of course, each of the preimage superposition sets is orthogonal within the set: $\langle\psi_i|\psi_j\rangle = \langle\widehat\psi_i|\widehat\psi_j\rangle = 0 \quad \forall i\ne j$.

We can therefore partition the Hilbert space into $2^n - 1$ orthogonal planes, each of which is spanned by a $|\psi_i\rangle$ and its corresponding $|\widehat\psi_i\rangle$ (or $|\psi_{z\rightarrow i}\rangle$), as well as a remaining space orthogonal to all those planes. 
With this perspective, the change of basis that differentiates between $\mathcal{C}_{\text{impostor}}$ and $\widehat{\mathcal{C}}_{\text{impostor}}$ can be described as a small rotation of angle $\theta_i$ in each of these planes and the identity in the remaining space.
\begin{align*}
    \mathcal U_3 := I 
    &- \sum_i \big( | \psi_i \rangle\langle \psi_i | + | \psi_{z\rightarrow i} \rangle\langle \psi_{z\rightarrow i} | \big) \\
    &+ \sum_i \big(\cos (\theta_i) |\psi_i\rangle + \sin (\theta_i) |\psi_{z\rightarrow i}\rangle \big)  \langle\psi_i| + \big(-\sin (\theta_i) |\psi_i\rangle + \cos (\theta_i) |\psi_{z\rightarrow i}\rangle \big) \langle\psi_{z\rightarrow i}|
\end{align*}

Then, 
$$I \otimes \mathcal{C}_{\text{impostor}} \ket{0}\ket{x}\ket{y} = \mathcal{U}_1\mathcal{U}_2\mathcal{U}_1 \ket{0}\ket{x}\ket{y} \; \text{ and } \;
\widehat{\mathcal{C}}_{\text{impostor}} = \mathcal{U}_1 (\mathcal{U}^{\dagger}_3 \otimes \mathcal{U}^{\dagger}_3) \mathcal{U}_2 (\mathcal{U}_3 \otimes \mathcal{U}_3) \mathcal{U}_1$$ 

It therefore suffices to show that $\mathcal{U}_3$ cannot be distinguished from the identity except with negligible advantage.
Specifically, we want to show that the eigenvalues of $I - \mathcal{U}_3$ are all negligible.
That's because if the magnitudes of all the eigenvalues of $I - \mathcal{U}_3$ are bounded from above by a negligible function $\varepsilon$,
then given any quantum state $|\phi\rangle$ before the application of $\mathcal{U}_3$ or $I$, 
and any subsequent transformation,
we have that the resulting Euclidean distance is $\big\||\phi\rangle - \mathcal{U}_3 |\phi\rangle\big\| = \big\| (I - \mathcal{U}_3) |\phi\rangle\big\| \le \varepsilon$, 
and thus by Theorem~\ref{thm:Ben+97_3.1}, when replacing $I$ with $\mathcal{U}_3$, the probability of success can change by at most $4\varepsilon$.

Since $I - \mathcal{U}_3$ acts independently on and maintains the $2^n - 1$ orthogonal planes, it suffices to look at each plane individually. Specifically, its non-zero eigenvalues come in pairs of magnitude
\begin{align*}
    |\lambda_i| &= |1-e^{\pm \textbf{i} \theta_i}| \\
    &= |1-\cos(\theta_i) \mp \textbf{i}\sin(\theta_i)| \\
    &= \sqrt{(1-\cos(\theta_i))^2 + \sin^2(\theta_i)} \\
    &= \sqrt{2(1 -  \cos(\theta_i))}  \\
    &= \sqrt{2\left(1 -  \langle \psi_i | \widehat\psi_i \rangle\right)}
\end{align*}

% Let $k_i$ and $\widehat k_i$ be the number of preimages of $i$ in 
% $H_{\text{impostor}}$ and $H_{\text{private}}$ respectively.
% % $\big|\big\{x \in \{0,1\}^m \big| H_{\text{impostor}}(x) = i\big\}\big|$
% Note that $\widehat k_i = k_i + k_{z \rightarrow i}$, where $k_{z \rightarrow i}$ is the number of preimages of the target image, $z$, in $H_{\text{impostor}}$ that were reassigned to image $i$ in $H_{\text{private}}$.

% \begin{align*}
%     \langle \psi_i | \widehat\psi_i \rangle 
%     &= \left( \frac{1}{\sqrt{k_i}} \sum_{x | H_{\text{impostor}}(x) = i} \langle x | \right)
%     \left( \frac{1}{\sqrt{\widehat k_i}} \sum_{x | H_{\text{private}}(x) = i} | x \rangle \right) \\
%     &= \frac{k_i}{\sqrt{k_i \cdot \widehat k_i}}
%     \\
%     &= \sqrt{\frac{k_i}{\widehat k_i}}
%     \\
%     &= \sqrt{\frac{k_i}{k_i + k_{z \rightarrow i}}}
%     \\
%     &= \sqrt{\frac{k_i + k_{z \rightarrow i} - k_{z \rightarrow i}}{k_i + k_{z \rightarrow i}}}
%     \\
%     &= \sqrt{1 - \frac{k_{z \rightarrow i}}{k_i + k_{z \rightarrow i}}}
% \end{align*}

In order to further break this down, let $k_i$ be the number of preimages of $i$ in 
$H_{\text{impostor}}$ and let $k_{z \rightarrow i}$ be the number of preimages of the target image, $z$, in $H_{\text{impostor}}$ that were reassigned to image~$i$ in $H_{\text{private}}$. We evaluate the inner product as
\begin{align*}
    \langle \psi_i | \widehat\psi_i \rangle 
    &= \left( \frac{1}{\sqrt{k_i}} \sum_{x | H_{\text{impostor}}(x) = i} \langle x | \right)
    \left( \frac{1}{\sqrt{k_i + k_{z \rightarrow i}}} \sum_{x | H_{\text{private}}(x) = i} | x \rangle \right)
    \\
    &= \sqrt{\frac{k_i}{k_i + k_{z \rightarrow i}}}
    % \\
    % &= \sqrt{\frac{k_i + k_{z \rightarrow i} - k_{z \rightarrow i}}{k_i + k_{z \rightarrow i}}}
    % \\&
    = \sqrt{1 - \frac{k_{z \rightarrow i}}{k_i + k_{z \rightarrow i}}}
    % \\&
    \ge 1 - \frac{k_{z \rightarrow i}}{k_i + k_{z \rightarrow i}}
\end{align*}
which gives
\begin{align*}
    |\lambda_i| 
    &= \sqrt{2\left(1 - \langle \psi_i | \widehat\psi_i \rangle\right)} 
    % \\&
    % = \sqrt{2\left(1 - \sqrt{1 - \frac{k_{z \rightarrow i}}{k_i + k_{z \rightarrow i}}}\right)}
    % \\&
    % \le \sqrt{2\left(1 - \left(1 - \frac{k_{z \rightarrow i}}{k_i + k_{z \rightarrow i}}\right)\right)}
    % \\&
    \le \sqrt{\frac{2 k_{z \rightarrow i}}{k_i + k_{z \rightarrow i}}}
\end{align*}

The following claim frames this bound in terms of $n$.

\begin{claim} \label{claim:Stage3-preimage-ratio}
    With overwhelming probability in the choice of $H$ and $H_{\text{private}}$, for all $i \in \{0,1\}^n\setminus\{z\}$,
    $$\frac{k_{z \rightarrow i}}{k_i + k_{z \rightarrow i}} \le 72n \cdot 2^{-n}$$
\end{claim}

% % Chernoff bounds used here
% $$P[X \ge (1 + \delta)\mu] \le e^{-\frac{\delta^2}{2+\delta}\mu} \quad \text{ for all } \delta > 0$$
% $$P[X \le (1 - \delta)\mu] \le e^{-\frac{\delta^2}{2} \mu} \quad \text{ for all } 0 < \delta < 1$$
% $$P[|X - \mu| \ge \delta\mu] \le 2 e^{-\frac{\delta^2}{3} \mu} \quad \text{ for all } 0 < \delta < 1$$

\begin{proof}
We show that the following all happen with overwhelming probability:
\begin{enumerate}[a)]
    \item for all $i \in \{0,1\}^n\setminus\{z\}$, \; $k_i > \frac12 \cdot 2^{m-n}$
    \item $\frac{1}{2} \cdot 2^{m-n} < k_z < 3 \cdot 2^{m-n}$
    \item for all $i \in \{0,1\}^n\setminus\{z\}$, \; $k_{z \rightarrow i} < 36n \cdot 2^{m-2n}$
\end{enumerate}

First we show that with overwhelming probability, for all $i \in \{0,1\}^n\setminus\{z\}$, \; $k_i > \frac12 \cdot 2^{m-n}$.
The expected number of preimages of any image $i$ is 
$\mathbb E[k_i] = 2^{m-n}$.
By a Chernoff bound,
$P[k_i \le \frac12 (2^{m-n})] \le e^{-\frac18 \cdot 2^{m-n}}$ for any particular image $i$.
By a union bound over the $2^n - 1$ images, 
the probability that for any $i$, 
$k_i \le \frac12 (2^{m-n})$, is at most 
$2^n \cdot e^{-\frac18 \cdot 2^{m-n}} \le e^{-(\frac18 \cdot 2^{m-n}-n)}$, 
which is negligible in $n$ as we have that $m \ge 2n$.

We next bound the number of preimages of the target image $z$. Specifically, we show that $\frac{1}{2} \cdot 2^{m-n} < k_z < 3 \cdot 2^{m-n}$.
The lower bound is identical to the one above for the other $k_i$'s.
The upper bound is given by another Chernoff bound as 
$P[k_z \ge 3 (2^{m-n})] \le e^{-2^{m-n}}$, which is likewise negligible in $n$.

Finally, we bound the number of preimages of $z$ in $H_{\text{impostor}}$ that can be mapped to any one $i$ in $H_{\text{impostor}}$. 
Specifically, we show that for all $i \in \{0,1\}^n\setminus\{z\}$, \; $k_{z \rightarrow i} < 36n \cdot 2^{m-2n}$.
Since we just showed that with overwhelming probability, $z$ has at least $\frac{1}{2} \cdot 2^{m-n}$ and at most $3 \cdot 2^{m-n}$ preimages, 
the expected number of these preimages distributed to each of the $2^n - 1$ other images is bounded by
$\frac12 \cdot 2^{m-2n} < \mathbb E[k_{z \rightarrow i}] < 6 \cdot 2^{m-2n}$.
By a Chernoff bound,
$P[k_{z \rightarrow i} \ge 6n (6 \cdot 2^{m-2n})] \le e^{- \frac{25n^2}{2+5n} \cdot \frac12 \cdot 2^{m-2n}} \le e^{- \frac32 n \cdot 2^{m-2n}}$ for any particular image $i$.
As before, by a union bound over the $2^n - 1$ images, 
the probability that for any $i$, 
$k_{z \rightarrow i} \ge 36n \cdot 2^{m-2n}$ is at most 
$2^n \cdot e^{- \frac32 n \cdot 2^{m-2n}} \le e^{- (\frac32 n \cdot 2^{m-2n} - n)}$
which is negligible in $n$ as $m \ge 2n$.

Putting these three things together, by a union bound over the three above events, with all but a negligible probability in $n$, for all $i$,
\begin{align*}
    \frac{k_{z \rightarrow i}}{k_i + k_{z \rightarrow i}} \le \frac{36n \cdot 2^{m-2n}}{\frac12 \cdot 2^{m-n}} = 72n \cdot 2^{-n}
\end{align*}
\end{proof}

We therefore get an upper bound of $\varepsilon: = 12\sqrt{n} \cdot 2^{-n/2}$ on the eigenvalues of $I - \mathcal{U}_3$, which is negligible in $n$, and therefore, as shown above, an upper bound of $4\varepsilon$ on the change in success probability incurred by replacing $I$ with $\mathcal{U}_3$.

We now use a standard hybrid argument over the at most $4q$ locations where $\mathcal{U}_3$ might appear. 
We start with $R''$, for which all such locations have the identity, and for which the success probability is the original success probability of $R$, 
namely $\eta$. 
One at a time, we insert a $\mathcal{U}_3$ at each location, each time incurring a loss of at most $4\varepsilon$ in the success probability. 
% (Note that since we showed this for all $|\phi\rangle$, 
% it does not matter what the quantum state is before applying $\mathcal{U}_3$.) 
With all $4q$ applications of $\mathcal{U}_3$, 
we therefore get a success probability $\eta'$ of at least $\eta - 16q\varepsilon - \gamma$ (where $\gamma$ is an additional negligible loss from the negligible chance that the sampled $H$ and $H_{\text{private}}$ are not covered by Claim~\ref{claim:Stage3-preimage-ratio}).

We therefore construct $R'$ in this way as a quantum oracle algorithm with advice with query access to the original random oracle $H$ and a $z$-cloning oracle $\mathcal C_z$. 
It simulates $R$ and redirects its oracles queries: 
Whenever $R$ makes a random oracle query, it redirects the query to its own simulated $H_{\text{impostor}}$, which makes at most a single query to $H$. 
Whenever $R$ makes a cloning oracle query, it redirects the query to its $\widehat{\mathcal{C}}_{\text{impostor}}$, 
which makes at most one query to $\mathcal C_z$ and two to $H$.
$R'$ thus satisfies the conditions of Corollary~\ref{cor:Stage2}, so its success probability $\eta' \ge \eta - 16q\varepsilon - \gamma$ must be negligible. 
Therefore, $\eta$, the success probability of $R$, must be negligible, thus completing the proof of Proposition~\ref{prop:Stage3}, and as a consequence, completing the proof of our main theorems, Theorem~\ref{thm:clonable_unreconstructible} and Theorem~\ref{thm:clonable_untelegraphable}.
\end{proof}

% \begin{corollary}
% Scheme~\ref{scheme:clonable_unreconstructible} is not reconstructible in the average case by an efficient quantum oracle algorithm.
% \end{corollary}
% \begin{corollary}
% Scheme~\ref{scheme:clonable_unreconstructible} is not telegraphable in the average case by efficient quantum oracle algorithms.
% \end{corollary}

\section{Implications for Complexity Theory}
\label{sec:complexity}

We now present an application of clonable-untelegraphable states to the study of complexity theory.
While there may be a number of possible connections to quantum complexity theory%
% \phantom{ }(for instance to state complexity classes [ref], the complexity of state preparation and state discrimination, etc.)%
, we focus on one that is of particular interest, which is to the computational no-go properties of efficiently verifiable quantum proofs. 
The longstanding open problem of whether the complexity classes $\QCMA$ and $\QMA$ are equal~\cite{aharonov2002quantum} asks whether classical proofs are just as powerful as quantum proofs in the setting of efficient quantum verification. 
Here, we investigate the power of quantum proofs which are not quite classical, but also not fully quantum, and are rather quantum states that violate some specific computational no-go property.

% We consider complexity classes 
% % that are described by non-interactive quantum proofs 
% in which the proofs being verified are quantum states which violate specific computational no-go properties. 
We first demonstrate that violating the efficient versions of either of the no-telegraphing or no-reconstruction properties makes the resulting complexity class equivalent to $\QCMA$, in which the proofs are classical strings. On the other hand, we show that this is not likely to be the case for the class $\clonableQMA$, in which the proofs are quantum states that are efficiently clonable. 
We justify this by giving a quantum oracle relative to which $\clonableQMA$ is not contained in $\QCMA$. We hope to inspire further investigation into the power of such quantum proofs.
Moreover, an in-depth understanding of the relative power of these complexity classes is important for constructing the cryptographic applications presented in Section~\ref{sec:cryptography}.

\subsection{Classical vs. Quantum Witnesses}

Recall the definitions of $\QCMA$ and $\QMA$ (Definitions~\ref{def:QCMA} and~\ref{def:QMA} respectively).
Note that the only difference between these two classes is the format of their witnesses: $\QMA$ allows any polynomial-sized quantum state as a witness, while $\QCMA$ restricts witnesses to be classical strings, or equivalently, restricts them to be in the computational basis.
It is evident that $\QCMA \subseteq \QMA$~\cite{aharonov2002quantum}.%
\footnote{\label{footnote:soundess-against-quantum-witnesses}%
As observed in~\cite{aharonov2002quantum}, the soundness condition of $\QCMA$ in Definition~\ref{def:QCMA} can be replaced with the one of $\QMA$ in Definition~\ref{def:QMA} against general quantum purported witnesses without any effect on the class. That's because the verifier can always force a quantum purported witness to be classical by measuring it in the computational basis. In other words, $\QCMA$ is also sound against \emph{quantum} witnesses.}
That is, the power of the class $\QMA$ is made no greater, and possibly weaker, by restricting its witnesses to be classical. Whether or not these two complexity classes are in fact equal has been a major open problem in quantum complexity theory since it was first posed in~\cite{aharonov2002quantum} over two decades ago. A sequence of works has shown increasingly strong oracle separations between the two classes, beginning with quantum oracle separations~\cite{aaronson2007quantum, fefferman2018quantum}, and most recently, separations by classical distributional oracles~\cite{natarajan2022classical, li2023classical}, but no separation relative to a standard classical oracle is yet known. Since both classes contain $\MA$ and $\NP$ and are contained in $\PP$ and therefore $\PSPACE$~\cite{marriott2004quantum}, a separation in the standard model (without reference to oracles) would imply separations which are not thought to be possible with existing techniques. The overriding question is nevertheless easy to phrase: \emph{Are classical witnesses as powerful as quantum witnesses in the context of efficient verification?}

In this section, we make progress on this question in a new direction: inspired by the new concept of clonable-untelegraphable states, we introduce a new complexity class, $\clonableQMA$, which sits in-between $\QCMA$ and $\QMA$, and we motivate the conjecture that it is not equal to either. This comes from considering weaker restrictions on the witnesses of $\QMA$. 
Instead of allowing the witnesses to be fully quantum as in $\QMA$ or restricting them to be fully classical as in $\QCMA$, we require them to violate specific computational no-go properties. We show that restricting the witnesses of $\QMA$ to be either efficiently reconstructable or efficiently telegraphable collapses the resulting class down to $\QCMA$. In other words, $\QCMA$ can be given an equivalent definition as the class $\QMA$ with efficiently reconstructable or efficiently telegraphable quantum witnesses. On the other hand, restricting the witnesses to be efficiently clonable does not have the same effect. In fact, as a consequence of the proof of our black-box separation between efficiently clonable and efficiently telegraphable quantum states, we give a quantum oracle black-box separation between $\QCMA$ and the new class, $\clonableQMA$, of $\QMA$ problems with efficiently clonable witnesses. 
Moreover, we argue without a formal proof that $\clonableQMA$ is not likely to equal $\QMA$ either, as this would imply the unlikely consequence of all $\QMA$-complete problems having efficiently clonable witnesses, which could prove to be a significant barrier to public-key quantum money. The class $\clonableQMA$ may therefore be a new complexity class standing strictly in-between $\QCMA$ and $\QMA$. We end Section~\ref{sec:complexity} by giving a candidate oracle-free problem in $\clonableQMA$ which may separate it from $\QCMA$, and we show that any such problem immediately yields back a set of states that is clonable but not efficiently telegraphable.

\subsection{Computational No-go Properties of Quantum Witnesses}

To motivate the discussion that follows, we start by giving a definition of $\QMA$ with efficiently reconstructable quantum witnesses, and then show that it is in fact an alternate definition of $\QCMA$. 

\begin{definition}[alternative definition of $\QCMA$ in terms of efficiently reconstructable witnesses]\label{def:QCMA'}
A decision problem $\mathcal{L} = (\mathcal{L}_{\YES}, \mathcal{L}_{\NO})$ is in $\QCMA'(c, f, s)$ if there exists a polynomial time quantum verifier $V$, a polynomial time quantum reconstructor $R$, and a polynomial $p$, such that
\begin{itemize}
    \item \textbf{Completeness:} if $x \in \mathcal{L}_{\YES}$, then there exists a \textbf{quantum} witness $\ket \psi$ on $p(|x|)$ qubits such that $V$ accepts on input $\ket x \ket \psi$ with probability at least $c$, and such that\\
    \textbf{Reconstruction Fidelity:} for this same $\ket\psi$, there exists classical advice string\footnote{As in the definition of Reconstruction (Definition \ref{def:reconstruction}) this advice string is trusted and instance-dependent. That is, we require that for this reconstructor $R$, and witness state $\ket{\psi}$, such a string \emph{exists} which allows $R$ to produce $\ket{\psi}$.} $a \in \bits^{p(|x|)}$ such that $R(a)$ succeeds at reconstructing $\ket \psi \bra \psi$ with fidelity at least $f$. \\
    That is, $\bra \psi R(a) \ket \psi \ge f$.
    \item \textbf{Soundness:} if $x \in \mathcal{L}_{\NO}$, then for all \textbf{quantum} states $\ket{\psi^*}$ on $p(|x|)$ qubits, $V$ accepts on input $\ket{x}\ket{\psi^*}$ with probability at most $s$.
\end{itemize}
\end{definition}

\begin{remark}\label{remark:qcma-sound-against-quantum}
Note that we only require the collection of valid witnesses to be efficiently reconstructable. That is, while the collection of witnesses for $\YES$-instances must be efficiently reconstructable, on the other hand, as with the standard definition of $\QCMA$
(see Footnote~\ref{footnote:soundess-against-quantum-witnesses} above), 
the class is sound against \textbf{any} quantum witness.
\end{remark}

\begin{theorem}\label{thm:qcma-definition-equivalence}
$\QCMA'(\frac9{10}, \frac9{10}, \frac1{10}) = \QCMA$.
That is, this definition of $\QCMA$ in terms of efficiently reconstructable witnesses is equivalent to the definition of $\QCMA$ in Definition~\ref{def:QCMA} in terms of classical witnesses, and describes the same class of decision problems.
\end{theorem}

\begin{proof}
The main idea is that classical witnesses are themselves efficiently reconstructable, and any efficiently reconstructable witness can be given as a classical witness instead. For completeness, we include a formal proof as follows.
For this proof, let the class described by Definition~\ref{def:QCMA'} be called $\QCMA'$ to distinguish it from that described by Definition~\ref{def:QCMA}.

\paragraph{$\QCMA \subseteq \QCMA'(\frac9{10}, \frac9{10}, \frac1{10})$:}
Let $\mathcal{L}$ be a decision problem in $\QCMA$. 
We have that there exists polynomial time quantum verifier $V$ which has completeness $\frac{9}{10}$ and soundness $\frac{1}{10}$.
Let $V'$ be the verifier that projects the witness onto the computational basis and then passes the result to $V$, and let $R'$ be the trivial reconstructor $R'(a) = \ket{a}\bra{a}$.
We show that $V'$ and $R'$ are a valid verifier and reconstructor pair for $\mathcal{L}$ in $\QCMA'$.

If $x \in \mathcal{L}_{\YES}$, then the guarantee of $\QCMA$ is that there exists a classical witness $w$ that causes $V$ to succeed with probability at least $\frac{9}{10}$. Therefore, let $\ket{w}$, the computational basis state corresponding to $w$, be the quantum witness for the $\QCMA'$ verifier $V'$. Furthermore, let $a=w$ be the advice for the reconstructor $R'$. 
Since $\ket w$ is already in the computational basis, we see that $V'$ accepts on input $\ket{x}\ket{w}$ with probability $\frac{9}{10}$, and furthermore, $R'(a) = R'(w) = \ket{w}\bra{w}$ with fidelity $1$.

If $x \in \mathcal{L}_{\NO}$, then for every classical string $w^*$, $V$ accepts the witness $\ket{w^*}$ with probability at most $\frac{1}{10}$, so $V'$, which first projects onto the computational basis, accepts any quantum witness $\ket {\psi^*}$ with the same probability.

Both $V'$ and $R'$ run in polynomial time, and the witness and advice string are the same length as the original witness. So $\mathcal{L} \in \QCMA'(\frac9{10}, \frac9{10}, \frac1{10})$.

\paragraph{$\QCMA'(\frac9{10}, \frac9{10}, \frac1{10}) \subseteq \QCMA$:}
Let $\mathcal{L}$ be a decision problem in $\QCMA'$. Then there exists a polynomial time quantum verifier $V'$ and a polynomial time reconstructor $R'$. Let $V$ be the $\QCMA$ verifier given by the composition of $V'$ and $R'$. Specifically, $V$ first passes its classical witness to the reconstructor $R'$, and then passes the result of that as a quantum witness to $V'$ (that is, $V_x(a) = V'_x(R'(a))$).
We show that $V$ is a valid verifier for $\mathcal{L}$ in $\QCMA$.

If $x \in \mathcal{L}_{\YES}$, then there exists a quantum witness $\ket{\psi}$ such that $V'$ accepts it with probability at least $\frac{9}{10}$, and there exists a classical advice string $a$ such that $\bra \psi R'(a) \ket \psi \ge \frac{9}{10}$. Without loss of generality, $V'$ can be seen as a projective measurement which succeeds with probability at least $\frac{9}{10}$.
Then by Lemma~\ref{lemma:composition}, $V$, which passes $a$ to $R'$ and the result of that to $V'$, accepts on input $x$ and witness $a$ with probability at least $\left(\frac{9}{10}\right)^2 - 2 \sqrt{\frac{1}{10^2}} = 0.61$.

If $x \in \mathcal{L}_{\NO}$, then for all quantum witnesses $\ket{\psi^*}$, $V'$ accepts with on input $\ket{x}\ket{\psi^*}$ probability at most $\frac{1}{10}$. Since $V$ produces the output of $V'$ on some quantum state, or a probabilistic mixture of quantum states, then $V$ will likewise only ever accept with probability at most $\frac{1}{10}$.

Since we have a constant soundness-completeness gap (between soundness $0.10$ and completeness $0.61$), we can amplify the gap by parallel repetition for $\QCMA$ to get the completeness and soundness required of Definition~\ref{def:QCMA}. $V$ runs in polynomial time, and the witness $a$ is the same length as the reconstruction advice. So $\mathcal{L} \in \QCMA$.

Combining these, we get that $\QCMA'(\frac9{10}, \frac9{10}, \frac1{10}) = \QCMA$.
\end{proof}

We define $\QMA$ with efficiently \emph{telegraphable} witnesses in the same way as in Definition~\ref{def:QCMA'}, but with both a polynomial time quantum deconstructor $D$ as well as the reconstructor $R$, which must together succeed at telegraphing the witness $\ket\psi$ with fidelity at least~$f$.

\begin{corollary}
$\QMA$ with telegraphable witnesses is also equal to $\QCMA$.
\end{corollary}
\begin{proof}
From Theorem~\ref{thm:telegraphing-implies-reconstruction}, we know that telegraphable witnesses are specifically also reconstructable, and with at least the same fidelity. From Theorem~\ref{thm:qcma-definition-equivalence}, we know that this makes this class a subset of $\QCMA$. In the other direction, classical witnesses are trivially telegraphable, which makes this class a superset of $\QCMA$. So the two classes are in fact equivalent.
\end{proof}

\begin{remark}
We see from this that we can define $\QCMA$ in three alternative but equivalent ways:
\begin{enumerate}
    \item as $\QMA$ with the collection of valid witnesses restricted to \textbf{classical strings}, or equivalently, quantum states in the computational basis
    \item as $\QMA$ with the collection of valid witnesses restricted to be \textbf{efficiently reconstructable} quantum states
    \item as $\QMA$ with the collection of valid witnesses restricted to be \textbf{efficiently telegraphable} quantum states
\end{enumerate}
\end{remark}

\begin{remark}
The task of separating between $\QMA$ and $\QCMA$ can thus be reframed as the task of finding decision problems in $\QMA$ for which any collection of witnesses is not efficiently reconstructable and/or not efficiently telegraphable.
\end{remark}

The violation of each computational no-go property -- that is, efficient reconstructability, efficient telegraphability, and efficient clonability -- brings collections of quantum states closer to being classical. 
It might then seem reasonable at this point to see a pattern and guess that every computational no-go violation by the witnesses of $\QMA$ would make it equal to $\QCMA$. 
That is, to guess that any such \emph{classicizing} restriction on the witnesses makes the quantum witnesses effectively only as good as classical witnesses.
We now give evidence \emph{against} that notion, by giving a quantum oracle black-box separation as evidence that in this context of efficient verification, \emph{efficiently clonable} witnesses are more powerful than classical, efficiently reconstructable, or efficiently telegraphable witnesses.

\subsection{Clonable Witnesses and $\clonableQMA$}

% We have seen that the class $\QCMA$ can be alternately redefined as having quantum witnesses that are efficiently reconstructable or otherwise efficiently telegraphable. It might appear, therefore, that any restriction on the witnesses of $\QMA$ that cause them to approach classical states would collapse the class to $\QCMA$. Our main point in this section is to argue that this is likely not the case for cloning.

% If $\QCMA$ captures witnesses that are efficiently reconstructable, then is there a complexity class capturing witnesses that are efficiently clonable?
% We definite the complexity class $\clonableQMA$ to do just that.

We give the following definition for $\clonableQMA$, the class of $\QMA$ problems that have efficiently clonable witnesses.

\begin{definition}[$\clonableQMA$]\label{def:clonableQMA}
A decision problem $\mathcal{L} = (\mathcal{L}_{\YES}, \mathcal{L}_{\NO})$ is in $\clonableQMA(c, f, s)$ if there exists a polynomial time quantum verifier $V$, a polynomial time quantum cloner $C$, and a polynomial $p$, such that
\begin{itemize}
    \item \textbf{Completeness:} if $x \in \mathcal{L}_{\YES}$, then there exists a quantum witness $\ket \psi$ on $p(|x|)$ qubits such that $V$ accepts on input $\ket x \ket \psi$ with probability at least $c$, and such that\\
    \textbf{Cloning Fidelity:} when given this same witness, $\ket\psi$, as input, $C$ succeeds at producing two independent copies of $\ket\psi$
    % $\ket \psi \ket \psi \bra \psi \bra \psi$ 
    with fidelity at least $f$. 
    That is, $\bra \psi \otimes \bra \psi C\big(\ket\psi\bra\psi\big) \ket \psi \otimes \ket \psi \ge f$.
    \item \textbf{Soundness:} if $x \in \mathcal{L}_{\NO}$, then for all quantum states $\ket{\psi^*}$ on $p(|x|)$ qubits, $V$ accepts on input $\ket{x}\ket{\psi^*}$ with probability at most $s$.
\end{itemize}
\end{definition}

% As with Definition~\ref{def:QCMA'}, the definition of $\clonableQMA$ only requires the collection of valid witnesses to be efficiently clonable, and the class is therefore sound against \emph{any} purported quantum witness.
% We likewise also take $\clonableQMA = \clonableQMA(\frac9{10}, \frac9{10}, \frac1{10})$. Note also that in order for a problem $\mathcal{L}$ to be in $\clonableQMA$, it is not required for any specific verifier for $\mathcal{L}$ to accept a collection of efficiently clonable witnesses. Rather, it is only required that there exist \emph{some} verifier, and \emph{some} mapping from instances to witnesses, such that the collection of all these valid witnesses is an efficiently clonable collection.

As with Definition~\ref{def:QCMA'}, the definition of $\clonableQMA$ only requires the collection of valid witnesses to be efficiently clonable, and the class is therefore sound against \emph{any} purported quantum witness.
We take $\clonableQMA = \bigcup_{(1-f) \in \negl(n)}\clonableQMA(\frac9{10}, f, \frac1{10})$.%
\footnote{
Note that as with the other classes mentioned, it seems likely that the parameters here can be set arbitrarily within a wide range. 
If the completeness and soundness errors are to be reduced, this can be done 
using the strong error reduction technique of~\cite{marriott2004quantum}, at the cost of an appropriate loss to the cloning fidelity error.
% The technique involves evaluating the verifier alternatingly both forward and in reverse, interspersed with measurements of the output bit and the ancilla registers.
Interestingly, the fact that the witnesses are clonable also allows simulating parallel repetition without the need for additional copies of the witness. That is, the single witness can be cloned as many times as needed while incurring the fidelity error. Doing this allows bringing the completeness and soundness near 1 and 0, respectively, without affecting the cloning fidelity error.
However, reducing the cloning fidelity error is not always possible.
% a different challenge, but is likely possible when the original error is small. We therefore leave as an open problem finding an error reduction procedure for the cloner.
% In order to reduce the cloning fidelity error as well, it seems likely that additionally inserting cloner evaluations both forward and in reverse, interspersed with swap tests and measurements of the cloner's extra register, would do the trick. We leave as an open problem a detailed analysis of this procedure.
}
% Note that as with the other classes mentioned, the completeness and soundness parameters are arbitrary, since completeness and soundness errors can be reduced using the strong error reduction technique of~\cite{marriott2004quantum}, at a polynomial loss to the cloning fidelity error (see Appendix~\ref{appendix:clonableQMA-robustness} for details). However, it is not a priori clear if errors in the cloning fidelity can likewise be reduced.

There are of course other ways that one could conceivably define a complexity class whose witnesses are efficiently clonable quantum states. For instance, the class could be defined with a single polynomial time quantum process that both verifies and clones in a single shot. 
% Or we could condition the cloning fidelity on the verifier accepting. 
We could also allow the cloner to accept a description of the problem instance as an additional input. We show in Appendix~\ref{appendix:clonableQMA-robustness} that up to a polynomial loss in 
% the completeness-soundness gap and/or 
the cloning fidelity, these variations are really all the same 
definition.

Note as well that by this definition, in order for a problem $\mathcal{L}$ to be in $\clonableQMA$, it is not required for any specific verifier for $\mathcal{L}$ to accept a collection of efficiently clonable witnesses. Rather, it is only required that there exist \emph{some} verifier, and \emph{some} mapping from instances to witnesses, such that the collection of all these valid witnesses is an efficiently clonable collection.

\subsubsection{Relationship to $\QMA$ and $\QCMA$}

\begin{theorem}
$\QCMA \subseteq \clonableQMA \subseteq \QMA$
\end{theorem}

\begin{proof}
The idea is that classical witnesses are efficiently clonable, and efficiently clonable witnesses are a subset of all quantum witnesses.
Recall from Remark~\ref{remark:qcma-sound-against-quantum} that $\QCMA$ can be sound even against general \emph{quantum} purported witnesses. So let $\mathcal{L}$ be a decision problem in $\QCMA$, and let $V$ be a verifier for $\mathcal{L}$ that accepts classical witnesses and is sound against quantum purported witnesses. Then $V$ is a valid verifier for $\mathcal{L}$ in $\clonableQMA$, with the cloning operation being the classical copy operation on the classical witness. Moreover, any decision problem $\mathcal{L} \in \clonableQMA$ is trivially also in $\QMA$, since the same verifier for the decision problem in $\clonableQMA$ serves as a verifier for it in $\QMA$.
\end{proof}

\begin{remark}
This hierarchy of three complexity classes gives a neat picture of the power of quantum verification in terms of the computational no-go properties of their quantum witnesses, from $\QCMA$ (efficiently reconstructable, efficiently telegraphable) to $\clonableQMA$ (efficiently clonable) to $\QMA$ (fully quantum).
\end{remark}

We conjecture that both containments are strict. Of course, an unconditional separation in either direction would imply an unconditional separation between $\QCMA$ and $\QMA$, as well a number of other resulting separations up to $\P \ne \PSPACE$ (a separation that is believed to be true, but not thought to be possible to prove with existing techniques).

Nevertheless, to justify that this is an entirely new complexity class, we give evidence for both separations.
In Subsection~\ref{sec:qcma-vs-clonableqma}, we show that the same quantum oracles that we used to separate no-cloning from no-telegraphing -- and the same collection of quantum states -- also serves to give us a black-box separation between $\QCMA$ and $\clonableQMA$. 
We conclude it by giving a quantum oracle $\mathcal{O}$ relative to which $\QCMA^{\mathcal{O}} \ne \clonableQMA^{\mathcal{O}}$, demonstrating that any attempt to prove their equality must at least be quantumly non-relativizing.

Moreover, we argue informally that $\clonableQMA$ is not likely to be equal to $\QMA$ either. That is because if $\QMA$ were contained in $\clonableQMA$, this would mean that every $\QMA$-complete problem would have efficiently clonable quantum witnesses, which would at the very least be surprising. 
At worst, if the equivalence between the classes were established through a witness-isomorphic reduction
(a reduction in which witnesses for one problem are mapped by an efficiently computable transformation to witnesses for the other~\cite{fischer1995witness,lynch1978structure}), this would rule out many schemes for public-key quantum money, as schemes based on the hardness of problems in $\QMA$ could be broken by mapping their witnesses to those of a $\clonableQMA$ problem for which there is an efficient cloner.

\subsection{Separating $\QCMA$ from $\clonableQMA$ in the Black-Box Model}\label{sec:qcma-vs-clonableqma}

Recall the collection of oracles and set of states defined in Scheme~\ref{scheme:clonable_unreconstructible}. We now use these same oracles to give a black-box separation between $\clonableQMA$ and $\QCMA$.
The problem that we use to show the separation is the problem of distinguishing a dummy cloning oracle from one which clones a state from Scheme~\ref{scheme:clonable_unreconstructible}.

\newlang{\ROHC}{ROHC}
\begin{definition}[Randomized oracle hidden cloning problem]
\label{def:rohc}
The randomized oracle hidden cloning problem, $\ROHC$, is an oracle promise problem where the input is a random oracle $H : \bits^m \to \bits^n$, and a unitary quantum oracle $\mathcal{C}$ on $m$ qubits, and the problem is to decide whether 
\begin{itemize}
    \item $\mathsf{YES}$: $\mathcal{C}$ is the $z$-cloning oracle, $\mathcal{C}_z$, relative to $H$, for some $z \in \bits^n$ (see Definition~\ref{def:z-cloning-oracle})
    \item $\mathsf{NO}$: $\mathcal{C}$ acts as the identity on all quantum states
\end{itemize}
\end{definition}

\begin{proposition}\label{prop:rohc-in-clonableqma}
$\ROHC$ is in $\clonableQMA$ in the black-box model.
\end{proposition}
\begin{proof}
Let the cloner and verifier required by $\clonableQMA$ be defined as follows:

Let $C$ be the cloner which very simply passes its input to the cloning oracle, $\mathcal{C}$, and outputs the result. That is, $C$ is just a wrapper for the cloning oracle.

Let $V$ be the verifier does the following: 
Pass the witness $\ket\psi$ and $\ket\bot$ to $\mathcal{C}$ and perform a projective measurement on the second register on the subspace spanned by $\ket\bot$ against the subspace spanned by everything else. Accept if the result is not $\ket\bot$.
% First, take the witness, $\ket\psi$, and pass it through $H$, measuring the output to get $z$. 
% Take the remaining quantum state, $\ket{\psi'}$, and 
% pass $\ket{\psi'}\ket\bot$ to $\mathcal{C}$. 
% Take the resulting state, and measure the output of applying $H$ independently to each of its registers. Accept only if both evaluate to $z$.

We show that $V$ and $C$ satisfy the following three properties: 
\begin{enumerate}
    \item Given a $\mathsf{YES}$ instance, there is a witness $\ket\psi$ that $V$ accepts with probability $1$.
    \item The witness $\ket\psi$ for every such $\mathsf{YES}$ instance is cloned by $C$ with perfect fidelity.
    \item Given a $\mathsf{NO}$ instance, $V$ rejects every witness $\ket\psi$ with probability $1$.
\end{enumerate}

All three of these are straightforward:

For any $\mathsf{YES}$ instance, $\mathcal{C} = \mathcal{C}_z$ for some $z \in \bits^n$, so let the witness for this instance be $\ket{\psi_z}$. $C = \mathcal{C}_z$ successfully clones $\ket{\psi_z}$. Therefore, given $\ket{\psi_z}$ as a witness, the verifier will not measure $\ket\bot$, and will accept with probability 1. On the other hand, for any $\mathsf{NO}$ instance, $\mathcal{C}$ will be the identity, so the register measured by $V$ will always be $\ket\bot$ at the end, and will therefore always reject.
\end{proof}

\begin{proposition}[Corollary of Proposition~\ref{prop:adding-the-z-cloning-oracle}]\label{prop:rohc-notin-qcma}
$\ROHC$ is not in $\QCMA$ in the black-box model.
\end{proposition}
\begin{proof}
This is a direct consequence of Proposition~\ref{prop:adding-the-z-cloning-oracle}.
That is, suppose for the sake of contradiction that $\ROHC\in\QCMA$. $\ROHC$ then has a $\QCMA$ verifier, $V$. 
% Let the pair of oracles $(H, \mathcal{C})$ be an $\ROHC$ instance, which is either a $\YES$ instance ($\mathcal{C}$ is a $z$-cloning oracle relative to $H$ for some $z$) or a $\NO$ instance ($\mathcal C$ acts as the identity on all states). 
$V$ satisfies the conditions of Proposition~\ref{prop:adding-the-z-cloning-oracle}, which means that
for every $\YES$ instance, 
$(H, \mathcal{C})$, where $\mathcal{C}$ is a $z$-cloning oracle relative to $H$ for some $z$,
and for every polynomial length witness string $w$, 
the probability that $V$ accepts this $\YES$ instance must be negligibly close to the probability that it accepts the corresponding $\NO$ instance in which $\mathcal{C}$ is replaced with the identity oracle, which is a contradiction.
\end{proof}

We now convert the black-box separation into a formal oracle separation between $\clonableQMA$ and $\QCMA$, by using standard techniques adapted from~\cite{aaronson2007quantum}.

\begin{theorem}\label{thm:clonableqma-qcma-oracle-separation}
There exists a quantum oracle $\mathcal{O}$ relative to which $\clonableQMA^{\mathcal{O}} \ne \QCMA^{\mathcal{O}}$.
\end{theorem}
\begin{proof}
We let $\mathcal{L}$ be a random unary language, where for each $n$, we set $1^n$ to be in $\mathcal{L}$ with probability $\frac12$. We let the quantum oracle $\mathcal{O} = \{\mathcal{O}_n\}_n$ be defined such that: for every $n$ for which $1^n \in \mathcal{L}$, we set $\mathcal{O}_n = (H_n, \mathcal{C}_n)$, where $H_n : \bits^{3n} \to \bits^{n}$ is a classical function chosen uniformly at random, $\mathcal{C}_n$ is the $z$-cloning oracle relative to $H_n$ for a randomly chosen $z_n \in \bits^n$, and the oracles are grouped together into a single quantum oracle; and likewise, for every $n$ for which $1^n \notin \mathcal{L}$, we set $\mathcal{O}_n = (H_n, I_n)$, where $H_n$ is chosen at random the same as before, and $I_n$ is an oracle which acts as the identity on $n$ qubits.

$\mathcal{L}$ is always in $\clonableQMA^{\mathcal{O}}$ for all choices of $\mathcal{L}$ and $\mathcal{O}$, since as shown in Proposition~\ref{prop:rohc-in-clonableqma}, there is a polynomial time verifier that can always solve the $\ROHC$ problem with a witness that is cloned by its corresponding polynomial time cloner. If $1^n \in \mathcal{L}$, then the verifier given there will accept the witness $\ket{\psi_{z_n}}$, the uniform positive superposition over preimages of $z_n$ in $H_n$, with probability $1$, and furthermore, the cloner will clone it with perfect fidelity. If $1^n \notin \mathcal{L}$, then there is no witness that will cause the verifier to accept with non-zero probability.

We now show that $\mathcal{L} \notin \QCMA^{\mathcal{O}}$ with probability 1 over the choice of $\mathcal{L}$ and $\mathcal{O}$.
For a fixed $\QCMA$ machine $M$, let $E_n(M, \mathcal{L}, \mathcal{O})$ be the event that $M^{\mathcal{O}}$ succeeds at determining whether $1^n \in \mathcal{L}$. 
That is, either $1^n \in \mathcal{L}$ and there exists a polynomial-length witness $w$ that $M^{\mathcal{O}}$ accepts with probability $\frac{9}{10}$, or $1^n \notin \mathcal{L}$ and $M^{\mathcal{O}}$ rejects all polynomial-length witnesses with probability $\frac{9}{10}$. 
By Proposition~\ref{prop:rohc-notin-qcma}, we have that for sufficiently large $n$, $\Pr_{\mathcal{L}, \mathcal{O}}\left[ E_n(M, \mathcal{L}, \mathcal{O}) \right] \le \frac{1}{10}$.
Since $\mathcal{O}$ is independent for different input lengths, and since on input $1^n$, $M^{\mathcal{O}}$ can only query $\mathcal{O}$ on polynomial length queries, this 
gives us that for infinitely many $n$,
\begin{align*}
    \Pr_{\mathcal{L}, \mathcal{O}}\left[ E_n(M, \mathcal{L}, \mathcal{O}) \; \middle| \; E_1(M, \mathcal{L}, \mathcal{O}) \wedge E_2(M, \mathcal{L}, \mathcal{O}) \wedge \dots \wedge E_{n-1}(M, \mathcal{L}, \mathcal{O}) \right] \le \frac{1}{10}
\end{align*}
and therefore,
\begin{align*}
    \Pr_{\mathcal{L}, \mathcal{O}}\left[ E_1(M, \mathcal{L}, \mathcal{O}) \wedge E_2(M, \mathcal{L}, \mathcal{O}) \wedge \dots \right] = 0
\end{align*}

Since there is only a countably infinite number of $\QCMA$ machines as a consequence of the Solovay-Kitaev Theorem~\cite{kitaev1997quantum}, we have by union bound that 
\begin{align*}
    \Pr_{\mathcal{L}, \mathcal{O}}\left[ \exists M : E_1(M, \mathcal{L}, \mathcal{O}) \wedge E_2(M, \mathcal{L}, \mathcal{O}) \wedge \dots \right] = 0
\end{align*}
or in other words, that $\mathcal{L} \notin \QCMA^{\mathcal{O}}$ with probability $1$ over the choice of $\mathcal{L}$ and $\mathcal{O}$.

We can thus fix a language $\mathcal{L}$ and a quantum oracle $\mathcal{O}$ such that $\mathcal{L} \in \clonableQMA^{\mathcal{O}}$, but $\mathcal{L} \notin \QCMA^{\mathcal{O}}$.
\end{proof}

We have used the same set of states and oracles which we showed are clonable but untelegraphable to prove an oracle separation between $\clonableQMA$ and $\QCMA$. Is this a general pattern? That is, can any set of clonable-untelegraphable states yield a corresponding separation between $\clonableQMA$ and $\QCMA$? We believe so.
Note, however, that a special feature of our separation is that besides being clonable but untelegraphable, the states we used as witnesses are efficiently \emph{samplable}, which is crucial for cryptographic applications (see Section~\ref{sec:cryptography}).
Moreover, this separation demonstrates the difficulty of unconditionally proving that any set of states is efficiently clonable but not efficiently telegraphable, as any such proof will likely yield a corresponding complexity class separation.

\subsection{$\clonableQMA$ in the Standard Model}
In Section~\ref{sec:qcma-vs-clonableqma}, we showed that the randomized oracle hidden cloning problem (deciding whether a quantum oracle is the identity or clones some hidden state from among a set) gives a black-box separation between $\QCMA$ and $\clonableQMA$.
There is a natural way to convert this into a promise problem in the standard model (that is, without reference to oracles).

\newlang{\CHC}{CHC}
\begin{definition}[Circuit hidden cloning problem]
The circuit hidden cloning problem, $\CHC$, is a promise problem where the input is the description of a $\poly(n)$-sized quantum circuit $\mathcal{C}$ on $2n$ qubits, and the problem is to decide whether 
\begin{itemize}
    \item $\mathsf{YES}$: $\left|\bra{\psi}\bra{\psi} \mathcal{C} \ket{\psi}\ket{0}\right|^2 \ge 1-\varepsilon$ for some $\ket\psi$ orthogonal to $\ket0$ 
    \item $\mathsf{NO}$: $\left|\bra{\psi}\bra{\phi} \mathcal{C} \ket{\psi}\ket{\phi}\right|^2 \le \varepsilon$ for all $\ket\psi$ and $\ket\phi$ 
    % \item $\mathsf{NO}$: $\mathcal{C}$ acts as the identity on all quantum states
\end{itemize}
\end{definition}

For sufficiently small $\varepsilon$, $\CHC$ is in $\clonableQMA$ for the same reason that $\ROHC$ is in $\clonableQMA$ in the black-box model, with the cloning oracle replaced by the circuit $\mathcal{C}$. While $\CHC$ cannot be shown to be hard for $\QCMA$ machines without implying a series of separations (starting from $\QCMA \ne \clonableQMA$, and up to $\P \ne \PSPACE$), we motivate this by comparison to the corresponding oracle problem, and by analogy to similar complete problems for $\QCMA$ and $\QMA$~\cite{wocjan2003qcmacomplete,janzing2005non,bookatz2013qmacomplete}.

We note that any separation between $\clonableQMA$ and $\QCMA$ immediately yields back a set of states that is efficiently clonable but neither efficiently reconstructable nor efficiently telegraphable.

\begin{theorem}\label{thm:complexity-separation-to-nogo-separation}
Any decision problem that separates $\clonableQMA$ from $\QCMA$ can be converted into a set of states that is efficiently clonable but not efficiently reconstructable or efficiently telegraphable.
\end{theorem}
\begin{proof}
Suppose that $\clonableQMA \ne \QCMA$, and let $\mathcal{L}$ be a decision problem in\break $\clonableQMA \setminus \QCMA$. Since $\mathcal{L} \in \clonableQMA$, there is polynomial time verifier, $V$, and a polynomial time cloner, $\mathcal{C}$, such that for every $\YES$ instance $x$, there is a witness $\ket{\psi_x}$ on $\poly(|x|)$ qubits that is verified by $V$, and which $\mathcal{C}$ clones with high fidelity. 
Now, consider the set of quantum states $S_n = \{(x, \ket{\psi_x}) : x \in \mathcal{L}_{\YES}, |x| = n\}$.
By construction, this set of states is efficiently clonable with high fidelity for all $n$. On the other hand, it is neither efficiently reconstructable nor efficiently telegraphable, as otherwise, $\mathcal{L}$ would be in $\QCMA$ by Definition~\ref{def:QCMA'} and Theorem~\ref{thm:qcma-definition-equivalence}.
\end{proof}

\section{Cryptographic Applications}\label{sec:cryptography}

\newcommand{\gen}{{{\sf Gen}}}
\newcommand{\enc}{{{\sf Enc}}}
\newcommand{\dec}{{{\sf Dec}}}
\newcommand{\send}{{{\sf Send}}}
\newcommand{\rec}{{{\sf Receive}}}

\newcommand{\pk}{{{\sf pk}}}
\newcommand{\qsk}[1][]{{{\sf |{\sf sk}#1\rangle}}}

\newcommand{\aux}{{{\sf aux}}}

In this section, we describe the cryptographic primitive of a parallelizable but un-exfiltratable key, and we show how to build it from clonable-untelegraphable states and a few other assumptions. For concreteness, we focus on the case of encryption.

We consider the following setup: a server has a secret key for a public key encryption scheme, which it uses to decrypt ciphertexts. Unfortunately, the server is compromised by a remote adversary. The adversary would like to exfiltrate the key, so that it can decrypt ciphertexts for itself. We imagine, however, that the server is only able to transmit \emph{classical} information, and we utilize a quantum secret key to prevent the key from being exfiltrated. This gives rise to the following definition:

\begin{definition}[Non-exfiltratable Encryption]\label{def:exfil} A \emph{non-exfiltratable public key encryption (nePKE) scheme} is a tuple of polynomial-time quantum algorithms $\gen,\enc,\dec$ such that:
\begin{itemize}
    \item $\gen(1^\lambda)$ samples a \emph{classical} public key $\pk$ and \emph{quantum} secret key $\qsk$.
    \item $\enc(\pk,m)$ takes as input the public key and a classical message $m\in\{0,1\}^\lambda$, and outputs a ciphertext that may be classical or quantum, which we denote as $c$ or $|c\rangle$, respectively.
    \item $\dec(\qsk,c)$ or $\dec(\qsk,|c\rangle)$ takes as input a secret key and a ciphertext, and outputs a classical message $m'$ and a new secret key $\qsk[']$ (the original secret key being consumed since it is a quantum state).
    \item {\bf Correctness.} For any polynomial $p(\lambda)$, there is a negligible function $\epsilon(\lambda)$ such that, for any sequence of messages $m_1,m_2,\cdots,m_{p(\lambda)}\in\{0,1\}^\lambda$, the following holds: \\
    Let $(\pk,\qsk[_0])\gets\gen(1^\lambda)$ and for each $i\in[p(\lambda)]$, let $(m_i',\qsk[_i])=\dec\big(\qsk[_{i-1}],\enc(\pk,m_i)\big)$.\\
    Then $\Pr\big[{m_i'=m_i} \; \forall i\in[p(\lambda)]\big]\geq 1 - \epsilon(\lambda)$.
    \item {\bf Non-exfiltration Security.} For any pair of quantum polynomial-time interactive algorithms $(\send,\rec)$, there exists a negligible function $\epsilon(\lambda)$ such that for each $\lambda$ and any pair of messages $m_0,m_1$, $|W_0-W_1|\leq\epsilon(\lambda)$ where $W_b$ is the probability $\rec$ outputs $b$ in the following experiment:
    \begin{itemize}
        \item Run $(\pk,\qsk)\gets\gen(1^\lambda)$ and give $\qsk,\pk$ to $\send$.
        \item $\send$ produces a classical string $u$.
        \item Compute $c\gets\enc(\pk,m_b)$ (resp. $|c\rangle$ if allowing for quantum ciphertexts)
        \item Run $\rec$ on $(\pk,u,c)$ (resp. $(\pk,u,|c\rangle)$) to get a bit $b'$.
    \end{itemize}
    Above, $\send$ plays the role of the compromised server, and $\rec$ the role of the remote attacker.
\end{itemize}
\end{definition}

\subsection{Parallelizeable Construction Using Clonable-Untelegraphable States}

We now show that clonable but untelegraphable states, along with appropriate cryptographic building blocks, yields un-exfiltratable encryption where keys can be copied, yielding a construction that facilitates parallelism.

\paragraph{Efficiently Samplable Clonable Witnesses.} First, we need a strengthening of $\clonableQMA \ne \QCMA$: we need that there is a decision problem $\mathcal{L} \in \clonableQMA \setminus \QCMA$ with efficiently samplable hard (instance, witness) pairs. That is, there is an efficient sampling procedure $\mathcal{S}(1^\lambda)$ which samples $\YES$ instances $x$ of size $\lambda$ along with their clonable witnesses $|\psi_x\rangle$, as well as $\NO$ instances $x$ (of course with no matching witness), such that in time polynomial in $\lambda$, it is infeasible to 
decide whether $x$ is a $\YES$ or $\NO$ instance when given $x$ and auxiliary classical information.
% compute any witness given just $x$ and auxiliary classical information.

\paragraph{Witness Encryption for $\QMA$.} We will also need the notion of \emph{extractable} witness encryption for $\QMA$, which we now define. For a candidate construction, we could use~\cite{ITCS:BarMal22}.

\begin{definition}[Extractable Witness Encryption for $\QMA$] An extractable witness encryption scheme for a decision problem $\mathcal{L} \in \QMA$ is a pair of efficient (potentially quantum) algorithms $(\enc,\dec)$ such that:
\begin{itemize}
    \item $\enc(1^\lambda,x,m)$ takes as input the security parameter, a $\QMA$ statement $x$, and a message $m$. It outputs a ciphertext that may be classical or quantum, which we denote as $c$ or $|c\rangle$, respectively.
    \item $\dec(x,c,|\psi\rangle)$ or $\dec(x,|c\rangle,|\psi\rangle)$ takes as input the $\QMA$ statement $x$, a classical or quantum ciphertext $c$ or $|c\rangle$, and a purported witness $|\psi\rangle$ for $x$.
    \item {\bf Correctness.} Let $R_\mathcal{L}(x)$ be the set of valid witnesses for $x$. Then for any $|\psi\rangle\in R_\mathcal{L}(x)$, we have that 
    $\Pr\left[\dec\big(x,\enc(1^\lambda,x,m),|\psi\rangle\big)=m\right]\geq 1-\negl(\lambda)$.
    \item {\bf Extractability security.} Consider an efficiently sampleable distribution $\mathcal{D}(1^\lambda)$ over instances $x$ and (potentially quantum) auxiliary information $\aux$. We say that $\mathcal{D}$ is \emph{hard} (as in, it is hard to extract a witness from $\aux$) if, for all quantum polynomial-time adversaries~$E$,
    $$\Pr_{(x,\aux)\gets\mathcal{D}}[E(x,\aux)\in R_\mathcal{L}(x)]\leq\negl(\lambda)$$

    We then say that $(\enc,\dec)$ is extractable if, for every quantum polynomial-time adversary $A$, every hard efficiently sampleable distribution $\mathcal{D}$, and every pair of messages $m_0,m_1$, 
    $$\left|\Pr_{(x,\aux)\gets\mathcal{D}}\left[A(x,\aux,\enc(1^\lambda,x,m_0))=1\right]-\Pr_{(x,\aux)\gets\mathcal{D}}\left[A(x,\aux,\enc(1^\lambda,x,m_1))=1\right]\right|\leq\negl(\lambda)$$
\end{itemize}
\end{definition}

% \subsection{Construction From Decision Problems}

% \begin{definition}
% We say that a average case decision problem $(\mathcal{L}, \mathcal{D})$ is hard for a complexity class $\mathcal{C}$ if for every Turing machine $M$, there exists a negligible function $\varepsilon$ such that $M$ fails to satisfy the conditions of the class with probability at least $\frac12 - \varepsilon$ over the distribution $\mathcal{D}$ over the instances.
% \end{definition}

% We say that an average case problem (L, D) is easy for a complexity class C if there exists a non-negligible delta and a machine M in the class that satisfies the conditions of the class with prob at least 1/2 + delta over the distribution on instances.
% In other words, it is hard for C if every machine M fails to satisfy the conditions with probability with prob at least 1/2 - epsilon (for some negligible epsilon).

\begin{theorem}\label{thm:non-exfiltration-sufficient-2}
Assuming both of the following, there exists un-exfiltratable encryption with clonable secret keys:
\begin{enumerate}
    \item A pair of efficiently samplable distributions, $\mathcal{D}_{\YES}$ and $\mathcal{D}_{\NO}$,  over $\YES$ instance-witness pairs and $\NO$ instances, respectively, of a problem 
    $\mathcal{L} \in \clonableQMA$ such that the average-case 
    problem $(\mathcal{L}, \frac12 \mathcal{D}_{\YES} + \frac12 \mathcal{D}_{\NO})$ is hard for $\QCMA$.
    % problem $(\mathcal{L}, \mathcal{D}')$ is hard for $\QCMA$, 
    % where $\mathcal{D}'$ is the distribution on instances that takes instances from $\mathcal{D}_{\YES}$ and $\mathcal{D}_{\NO}$ with equal probability.
    \item Extractable witness encryption for $\QMA$
\end{enumerate}
Here, $\frac12 \mathcal{D}_{\YES} + \frac12 \mathcal{D}_{\NO}$ is the distribution on instances that takes instances from $\mathcal{D}_{\YES}$ and $\mathcal{D}_{\NO}$ with equal probability.
% The distribution $\mathcal{D}$ here outputs an instance $x$ along with a witness $\ket{\psi_x}$ if $x$ is a $\YES$ instance or an empty register $\ket\bot$ if $x$ is a $\NO$ instance.
\end{theorem}

Note that the condition of being efficiently samplable does not contradict the hardness of the decision problem or the unreconstructability of the witnesses, as this only allows us to sample a witness for a random instance, and not for a specific one. Indeed, relative to an oracle,
% the Randomized Oracle Hidden Cloning problem (Definition~\ref{def:rohc})
Scheme~\ref{scheme:clonable_unreconstructible} 
satisfies the necessary requirements.%
\footnote{
We in fact need a slight modification to Scheme~\ref{scheme:clonable_unreconstructible} to make it a decision problem, by having the cloning oracle only clone half the valid states, whose images in the random oracle then become the $\YES$ instances. We can sample random instances (along with their potential witnesses) by measuring the output of the random oracle on a uniform superposition. From there it is straightforward to check if the instance sampled is a $\YES$ instance or a $\NO$ instance. However given just the instance, it is hard to produce a witness or even to tell if it is a $\YES$ instance or a $\NO$ instance (Propositions~\ref{prop:adding-the-z-cloning-oracle} and~\ref{prop:Stage3}).
}

%We observe that, relative to an oracle, Scheme~\ref{scheme:clonable_unreconstructible} satisfies the necessary requirements. Indeed, instances will be strings $z\in\{0,1\}$, and witnesses will be the corresponding states from the pre-image superposition set $S_H$. To sample an (instance, witnesses) pair, construct the uniform superposition over $\{0,1\}^m$, apply $H$ in superposition, and measure the output. The output will be $z$, and the resulting quantum state will collapse to the witness. 

%To verify, first check that the witness maps to $z$ under $H$. Then check that the state is cloned by $\mathcal{C}_H$. This is done by applying $\mathcal{C}_H$, and then checking that each output state also maps to $z$ under $H$. The uniform superpostition of pre-images of $z$ is the only state that passes this verification, meaning witnesses are unique.
\medskip
\noindent The construction is as follows:

\begin{proof}
We are given a decision problem $\mathcal{L} \in \clonableQMA$, a $\clonableQMA(c, f, s)$ verifier $V$ and cloner $C$ for $\mathcal{L}$, and polynomial time instance samplers $\mathcal{S}_{\YES}$ and $\mathcal{S}_{\NO}$ for $\YES$ instance-witness pairs and $\NO$ instances, respectively for $\mathcal{L}$.

Our un-exfiltratable scheme $(\gen,\enc,\dec)$ is defined as follows:
\begin{itemize}
    \item $\gen(1^\lambda)$ runs $(x,|\psi\rangle)\gets\mathcal{S}_{\YES}(1^\lambda)$ and outputs $\pk=x$ as the public key and $\qsk=|\psi\rangle$ as the secret key.
    \item $\enc(\pk,m)$ runs and outputs the result of $\enc'(1^\lambda,\pk,m)$
    \item $\dec(\qsk,c)$ runs and outputs the result of $\dec'(x,c,\qsk)$.
\end{itemize}
Correctness follows immediately from the correctness of $(\gen',\enc')$,
and the secret keys $\qsk$ are clonable by the cloner $C$ for $\mathcal{L}$.
% The clonability of $\qsk$ follows from the fact that the problem $\mathcal{L}$ is in $\clonableQMA$. 
It remains to prove security. Consider an adversary $(\send,\rec)$ breaking the un-exfiltratability of $(\gen,\enc,\dec)$. This means there is a non-negligible $\epsilon(\lambda)$ and a pair of messages $m_0,m_1$ such that $|W_0-W_1|\geq\epsilon(\lambda)$ where $W_0,W_1$ are the quantities in Definition~\ref{def:exfil}. We construct a distribution $\mathcal{D}'$ and algorithm $A'$ which attacks $(\enc',\dec')$:
\begin{itemize}
    \item $\mathcal{D}'(1^\lambda)$ runs $(x,|\psi\rangle)\gets\mathcal{S}_{\YES}(1^\lambda)$, and then runs $u\gets \send(x,|\psi\rangle)$, where $u$ is a classical string. It outputs $(x,\aux=u)$.
    \item $A'(x,\aux,c)$ runs $\rec(x,\aux,c)$ and outputs whatever $\rec$ outputs.
\end{itemize}
By construction, we have that 
$$\left|\Pr_{(x,\aux)\gets\mathcal{D}'}\left[A'(x,\aux,\enc'(1^\lambda,x,m_0))=1\right]- \Pr_{(x,\aux)\gets\mathcal{D}'}\left[A'(x,\aux,\enc'(1^\lambda,x,m_1))=1\right]\right|\geq\epsilon(\lambda)$$
Thus, by the extractability security of $(\enc',\dec')$, we have that $\mathcal{D}'$ must not be hard. 
This means there exists a quantum polynomial-time extractor $E$ and non-negligible $\delta(\lambda)$ such that \linebreak
$\Pr_{(x,\aux)\gets\mathcal{D}'}\big[\Pr[V(x,E(x,\aux)) = 1] \ge c\big]\geq\delta(\lambda)$, 
where $V$ is the $\clonableQMA$ verifier for $\mathcal{L}$ and $c$ is the completeness parameter. That is, with probability at least $\delta(\lambda)$, $E$ extracts a valid quantum witness for $V$.
%Since witnesses are unique, we must actually have $\Pr[\|<\psi|E(x,u)\|^2:(x,\aux)\gets\mathcal{D}]\geq\delta$. We then use $E$ to construct a telegrapher for the distribution $(x,|\psi\rangle)\gets\mathcal{S}$. The sender simply runs $\send(x,|\psi\rangle)$ and sends the resulting $u$. The receiver runs $E(x,u)$, which outputs the state $|\psi\rangle$ with non-negligible probability $\delta$.

We then use $E$ and $V$ to construct an average-case $\QCMA$ verifier $V'$ for the distribution of instances coming from the equal mixture of $\mathcal{S}_{\YES}$ and $\mathcal{S}_{\NO}$. 
With instance $x$ and witness $u$, $V'(x,u)$ is defined as $V'(x,u):=V(x,E(x,u))$. 
We then have that over the instance distribution of $\mathcal{S}_{\YES}$, the $u$ outputted by $\send(x,|\psi\rangle)$ is, with non-negligible probability at least $\delta(\lambda)$, a witness for $x$ relative to $V'$; in particular the witness exists. 
We also have that for instances sampled from $\mathcal{S}_{\NO}$, there is no quantum witness that $V$ accepts with probability greater than $s$, and there is therefore no classical witness $u$ that $V'$ accepts with probability greater than $s$. 
We have therefore that $V'$ satisfies the conditions of $\QCMA$ for $(\mathcal{L}, \mathcal{S})$ with probability $\frac12 + \frac12 \delta(\lambda)$ over the distribution $\mathcal{S}$ on instances induced by sampling equally from $\mathcal{S}_{\YES}$ and $\mathcal{S}_{\NO}$, contradicting the condition that $(\mathcal{L}, \mathcal{S})$ is hard for $\QCMA$.
We then have that, as claimed, $(\gen,\enc,\dec)$ is a secure un-exfiltratable encryption with clonable quantum secret keys.
\end{proof}

\bibliographystyle{alpha}
\bibliography{bib}

\appendix

% \section{Appendix}

\section{Proof of the No-Cloning and No-Telegraphing Theorems}
\label{appendix:no-cloning-equals-no-telegraphing}

The No-Cloning Theorem was discovered independently three separate times~\cite{Park70,WoottersZurek82,Dieks82}.
% but the version we present here is due to~\cite{YUEN1986405}. 
The No-Telegraphing Theorem can be seen as a corollary of the No-Cloning Theorem, and is due to~\cite{Werner_1998} (where it is referred to as the No-Teleportation Theorem). In both cases, the theorems state that cloning and telegraphing, respectively, are not possible for general unknown quantum states. The most precise statement of the theorems, as presented below, states that these two tasks are possible only on orthogonal collections of states.
Though the two theorems were historically discovered separately, we present them together to emphasize the direct connection between~them.

\begin{theorem}[No-Cloning Theorem and No-Telegraphing Theorem]
Let $\mathcal{H}$ be a Hilbert space, and let $S = \{\ket{\psi_i}\}_{i \in [k]}$ be a collection of pure quantum states on this Hilbert space.
The following are equivalent:
\begin{enumerate}
    \item $S$ can be perfectly cloned
    \item $S$ can be perfectly telegraphed
    \item $S$ is a collection of orthogonal states, with duplication \textnormal{(}$\forall i,j \; \left| \langle{\psi_i}|{\psi_j}\rangle \right|^2$ is either $0$ or $1$\textnormal{)}
\end{enumerate}
\end{theorem}

\begin{proof}
We show that $(2) \implies (1) \implies (3) \implies (2)$.

\paragraph{$(2) \implies (1)$:}
If $S$ can be telegraphed, then there is a sending process $\mathsf{Send}$ which when given $\ket{\psi_i}$, outputs a classical message $c_i$, and a receiving process, 
$\mathsf{Receive}$, which when given $c_i$, outputs the state $\ket{\psi_i}$. Run $\mathsf{Send}$ once on $\ket{\psi_i}$ to get $c_i$, and then run $\mathsf{Receive}$ twice independently on $c_i$ to output two independent and unentangled copies of $\ket{\psi_i}$.

% That is, there exist efficient quantum algorithms 
% $\mathsf{Send}(|\psi\rangle) \rightarrow c$
% and 
% such that for all $|\psi_z\rangle \in S$, 
% $|\phi\rangle := \mathsf{Receive}(\mathsf{Send}(|\psi_z\rangle))$ passes verification for $z$ with probability at least $\eta$.

\paragraph{$(1) \implies (3)$:} This is due to Yuen's proof of the No-Cloning Theorem~\cite{YUEN1986405}, which we present here.
The most general cloning process can be described by a unitary $U$ which acts as
$$\ket{\psi_i} \ket{\bot} \ket{\chi} \xrightarrow{U} \ket{\psi_i} \ket{\psi_i} \ket{\chi_i}$$
where $\ket\bot$ is the state of an empty register, and $\ket{\chi}$ and $\ket{\chi_i}$ are states of the ancilla registers and/or working memory of the cloner.

Because unitary operations preserve inner products, we get that for every $i$ and $j$,
\[
\braket{\psi_i|\psi_j}
= \braket{\psi_i|\psi_j}\braket{\bot|\bot}\braket{\chi|\chi} \\
% = \bra{\psi_i}\bra{\bot}\bra{\chi} U^\dag U \ket{\psi_j}\ket{\bot}\ket{\chi} \\
= \braket{\psi_i|\psi_j}\braket{\psi_i|\psi_j}\braket{\chi_i|\chi_j} \\
= \braket{\psi_i|\psi_j}^2\braket{\chi_i|\chi_j}
\]
% \begin{align*}
%     \braket{\psi_i|\psi_j}
%     &= \braket{\psi_i|\psi_j}\braket{\bot|\bot}\braket{\chi|\chi} \\
%     &= \braket{\psi_i|\psi_j}\braket{\psi_i|\psi_j}\braket{\chi_i|\chi_j} \\
%     &= \braket{\psi_i|\psi_j}^2\braket{\chi_i|\chi_j}
% \end{align*}
from which we have that either $\braket{\psi_i|\psi_j} = 0$ or $\braket{\psi_i|\psi_j}\braket{\chi_i|\chi_j} = 1$. In the second case, because both $\braket{\psi_i|\psi_j}$ and $\braket{\chi_i|\chi_j}$ can have magnitude at most 1, both magnitudes must equal 1. 
So we get that either $\left|\braket{\psi_i|\psi_j}\right| = 0$ or $\left|\braket{\psi_i|\psi_j}\right| = 1$.

\paragraph{$(3) \implies (2)$:} 
Since the states are all pairwise either orthogonal or identical to one another (up to an irrelevant overall phase), they define an orthonormal set, which can be completed to a full orthonormal basis $B = \{\ket{\phi_j}\}_j$ for the Hilbert space (removing any such redundant elements that are identical up to their overall phase). 
The projective measurement with operators $\Pi_j = \ket{\phi_j}\bra{\phi_j}$ therefore perfectly distinguishes the set $S$ (of course, without distinguishing between those identical elements, for which it is not necessary to distinguish). 

On input $\ket{\psi_i}$, the sending process therefore performs the projective measurement $\{\Pi_j\}_j$ to get outcome $j$ such that $\left|\braket{\psi_i | \phi_j}\right| = 1$, and sends this $j$. Then, the receiving process, on input $j$, prepares $\ket{\phi_j}$, which is equal to $\ket{\psi_i}$, up to an irrelevant overall phase, thus successfully telegraphing~$\ket{\psi_i}$.
\end{proof}

\section{Measuring an Approximation of a Quantum State}

The following lemma is useful if we expect to have a pure quantum state $\ket{\psi}$, but instead, we only have an approximate version of it, $\rho$, which may in general be a mixed state. We want that any measurement on $\ket\psi$ can be approximately performed on $\rho$ instead.

\begin{lemma}[Measuring an approximation of a state]
\label{lemma:composition}
Let $\ket \psi \in \mathcal{H}$ be a pure state,  
let $\Pi$ be a measurement operator for the binary projective measurement $\{\Pi, \mathbb I - \Pi\}$ on $\mathcal{H}$, 
and let $\rho \in \mathcal{D}(\mathcal{H})$ be a mixed state such that

\begin{enumerate}
    \item $\bra\psi\rho\ket\psi \ge p_1$
    \item $\bra\psi \Pi \ket\psi \ge p_2$
\end{enumerate}
Then $\mathsf{tr}(\Pi \rho) \ge p_1 p_2 - 2 \sqrt{(1-p_1)(1-p_2)}$.
\end{lemma}

This means that if $\rho$ is very close to $\ket\psi\bra\psi$, and the measurement given by $\Pi$ succeeds with good probability on $\ket\psi$, then the same measurement also succeeds on $\rho$, though with appropriately smaller probability.

\begin{proof}
The projector $\Pi$ can be expanded in its eigenbasis as $\Pi = \sum_{i=1}^{k} \ket{\pi_i}\bra{\pi_i}$, where $k = \rank(\Pi)$, and the vectors $\{\ket{\pi_i}\}_{i\in[k]}$ span the subspace onto which $\Pi$ projects.
Without loss of generality, $\ket{\pi_1}$ is the normalized projection of $\ket\psi$ onto that subspace. 
That is, $\ket{\pi_1} = \frac{\Pi \ket\psi}{\sqrt{\bra\psi\Pi\ket\psi}}$. 
Note that $\lvert\braket{\psi\vert\pi_1}\rvert^2 = \bra\psi \Pi \ket\psi \ge p_2$, and thus $\lvert\braket{\psi^\perp\vert\pi_1}\rvert^2 \le 1 - p_2$ for any $\ket{\psi^\perp}$ orthogonal to $\ket\psi$.

Now let $\ket\phi$ be any pure state, and set $\alpha$ and $\beta$
such that $\ket\phi = \alpha \ket\psi + \beta\ket{\psi^\perp}$ where $\ket{\psi^\perp}$ is some state orthogonal to $\ket\psi$. Then
% where $|\alpha|^2 \ge p_1$ and $|\beta|^2 \le 1 - p_1$.
\begin{align*}
    % \label{eq:composition-lemma}
    \mathsf{tr}(\Pi \ket\phi\bra\phi) 
    &= \bra\phi \Pi \ket\phi  \\
    &= \braket{\phi | \pi_1}\braket{\pi_1 |\phi} + \bra\phi \left( \sum_{i=2}^{k} \ket{\pi_i}\bra{\pi_i} \right) \ket\phi \\
    &\ge \braket{\phi | \pi_1}\braket{\pi_1 |\phi} \\
    &= \left(\alpha^* \braket{\psi | \pi_1} + \beta^* \braket{\psi^\perp | \pi_1}\right)\left(\alpha \braket{\pi_1|\psi} + \beta \braket{\pi_1|{\psi^\perp}}\right)  \\
    &= |\alpha|^2 \left|\braket{\psi | \pi_1}\right|^2 + |\beta|^2 \lvert\braket{\psi^\perp | \pi_1}\rvert^2 + \alpha^*\beta \braket{\psi | \pi_1}\braket{\pi_1|{\psi^\perp}} + \alpha\beta^* \braket{\psi^\perp | \pi_1}\braket{\pi_1|\psi} \\
    &\ge |\alpha|^2 \left|\braket{\psi | \pi_1}\right|^2 + |\beta|^2 \lvert\braket{\psi^\perp | \pi_1}\rvert^2 - 2|\alpha||\beta| \cdot \lvert\braket{\psi | \pi_1}\rvert \cdot \lvert\braket{{\psi^\perp}|\pi_1}\rvert \\
    &\ge |\alpha|^2 \left|\braket{\psi | \pi_1}\right|^2 - 2|\beta| \cdot \lvert\braket{{\psi^\perp}|\pi_1}\rvert \\
    &\ge |\alpha|^2 \; p_2 - 2|\beta|  \sqrt{1 - p_2}
\end{align*}

Now, let $\{\ket{\phi_i}\}_i$ be an eigenbasis for $\rho$ with corresponding eigenvalues $q_i$. Then $\rho = \sum_i q_i \ket{\phi_i}\bra{\phi_i}$. For each $i$, set $\alpha_i$ and $\beta_i$ such that $\ket{\phi_i} = \alpha_i \ket\psi + \beta_i \ket{\psi^\perp_i}$ for some $\ket{\psi^\perp_i}$ orthogonal to $\ket\psi$. We have

\begin{align*}
    \sum_i q_i |\alpha_i|^2 
    = \sum_i q_i |\langle\psi|\phi_i\rangle|^2 
    = \bra\psi\rho\ket\psi 
    \ge p_1 
\end{align*}
\begin{align*}
    \left(\sum_i q_i |\beta_i|\right)^2
    \le \sum_i q_i |\beta_i|^2
    = \sum_i q_i (1 - |\alpha_i|^2)
    =  1 - \sum_i q_i |\alpha_i|^2
    \le 1- p_1
\end{align*}

Putting this together, we get

\begin{align*}
    \mathsf{tr}(\Pi \rho) 
    &= \sum_i q_i \mathsf{tr}(\Pi \ket{\phi_i}\bra{\phi_i}) \\ 
    &\ge \sum_i q_i \left( |\alpha_i|^2 \; p_2 - 2|\beta_i|  \sqrt{1 - p_2} \right) \\
    &= \left(\sum_i q_i |\alpha_i|^2 \right) p_2 - 2 \left(\sum_i q_i |\beta_i| \right) \sqrt{1 - p_2} \\
    &\ge p_1 p_2 - 2 \sqrt{(1-p_1)(1-p_2)}
\end{align*}
\end{proof}

\section{Robustness of $\clonableQMA$ to Definitional Variations}
\label{appendix:clonableQMA-robustness}

We give some variations on the definition of $\clonableQMA$ (Definition~\ref{def:clonableQMA}) and show that these definitional variations do not change the class (up to a polynomial loss in the cloning fidelity).

\begin{proposition}
The class $\clonableQMA$ does not change if the cloner in Definition~\ref{def:clonableQMA} is allowed to take the description of the problem instance, $x$, as an extra input.
\end{proposition}
\begin{proof}
Here, the cloner receives the classical description of the instance, $x$, as well as the quantum witness, $\ket\psi$, that it must clone. One direction is clear, since the cloner can always disregard the instance description if it does not need it. 
In the other direction, if the witness does not already include the instance description, the witness can always be augmented to include the it as part of the witness. That is, let the new witness now be $\ket{\psi}\otimes\ket{x}$.
This clearly does not affect the ability of the verifier to verify the witness, but at the same time, it allows the cloner to implicitly receive the instance description as input.
\end{proof}

Note that a consequence of this is that a collection of valid witnesses for a verifier need not even be anywhere close to orthogonal to one another for the problem it verifies to be contained in $\clonableQMA$. 
Rather, there must be \emph{some} verifier and \emph{some} collection of witnesses for that verifier such that this collection is clonable.

As a concrete example, a $\clonableQMA$ problem $\mathcal{L}$, can have a verifier which, when $n=2$, accepts one of 4 different witnesses from the set $S = \{\ket0, \ket1, \frac{1}{\sqrt{2}}\left(\ket0 + \ket 1\right), \frac{1}{\sqrt{2}}\left(\ket0 - \ket 1\right)\}$, depending on the instance, $x \in \bits^2$. The set $S$ is clearly not clonable with very high fidelity. 
However, this is not a problem, because the verifier could instead accept a witnesses from the set $S' = \{{\ket0 \otimes \ket{00}}, {\ket1 \otimes \ket{01}}, {\frac{1}{\sqrt{2}}\left(\ket0 + \ket 1\right) \otimes \ket{10}}, {\frac{1}{\sqrt{2}}\left(\ket0 - \ket 1\right) \otimes \ket{11}}\}$, which serves the same purpose, is orthogonal, and can be cloned efficiently with fidelity $1$.

This same kind of transformation could in general be applied to the witnesses of any problem in $\QMA$ to make the set of witnesses an orthogonal set (and therefore an information-theoretically clonable one), which means that $\QMA$ itself could be defined in terms of orthogonal or (inefficiently) clonable witnesses. Thus, the question of whether $\QMA$ is contained in $\clonableQMA$ is therefore \emph{not} a question of whether problems in $\QMA$ can have witnesses that are clonable (they can, as we just showed), but rather whether they can have witnesses that are clonable \emph{efficiently}.

\paragraph{Combined Verifier-Cloner}
We give a second variation on the definition of $\clonableQMA$, in which the verifier and cloner are a combined procedure acting on a single witness. Note that in this case, the reduction comes with a loss in the parameters.%
\footnote{
In one of the directions, it incurs a loss in the cloning fidelity of $f \to c^2 f - 2 \sqrt{(1-c^2)(1-f)}$. 
We can also view this loss from the perspective of the error, taking $c = 1 - \varepsilon_c$ and $f = 1 - \varepsilon_f$, which gives us a resulting 
cloning fidelity error with a loss of at most $\varepsilon_f \to (2\varepsilon_c + \varepsilon_f) + (2 \sqrt{2})\sqrt{\varepsilon_c \phantom{\vert} \varepsilon_f}$. This shows that when the cloning fidelity error starts out as negligible, then the resulting error is negligible as well.
}

\begin{definition}[$\clonableQMA$ with a combined verifier-cloner]\label{def:clonableQMA-var2}
A decision problem $\mathcal{L} = (\mathcal{L}_{\YES}, \mathcal{L}_{\NO})$ is in $\clonableQMA'(c, f, s)$ if there exists a verifier-cloner $V_C$, a polynomial time quantum Turing machine outputting both an accept/reject bit as well as a quantum state on two registers, and a polynomial $p$, such that
\begin{itemize}
    \begin{samepage}
    \item \textbf{Completeness and Cloning Fidelity:} if $x \in \mathcal{L}_{\YES}$, then there exists a quantum witness $\ket \psi$ on $p(|x|)$ qubits such that on input $\ket x \ket \psi$, $V_C$ 
    \begin{itemize}
        \item accepts with probability at least $c$, and
        \item succeeds at producing two independent copies of $\ket\psi$ with fidelity at least $f$. 
    \end{itemize}
    \end{samepage}
    % $\ket \psi \ket \psi \bra \psi \bra \psi$ 
    % That is, $\bra \psi \otimes \bra \psi C\big(\ket\psi\bra\psi\big) \ket \psi \otimes \ket \psi \ge f$.
    \item \textbf{Soundness:} if $x \in \mathcal{L}_{\NO}$, then for all quantum states $\ket{\psi^*}$ on $p(|x|)$ qubits, $V_C$ accepts on input $\ket{x}\ket{\psi^*}$ with probability at most $s$.
\end{itemize}
\end{definition}

\begin{proposition}
% %
% \footnote{
% In order to show that these are the same definition without a loss in the reduction, we need a way of amplifying the cloning fidelity. This is not a trivial task, as the standard method of error reduction using parallel repetition does not suffice. That's because while giving $k$ copies of the witness may allow cloning a new $(k+1)$'th copy with higher fidelity than when given just a single copy, the definition requires making a copy of the full witness that is given to the verifier. So in this case, if the verifier is given $k$ copies of a state $\ket{\psi}$ as witness, it would have to produce a total of $2k$ copies of $\ket{\psi}$ with high fidelity, which is not necessarily an easier task. However, note that this kind of parallel repetition is a special case of an error-correcting code on the witness. Perhaps some more sophisticated error-correcting code 
% % or a strong error reduction method 
% would allow properly amplifying the cloning fidelity.
% }
$$\clonableQMA'(c, f, s) \subseteq \clonableQMA(c, f, s) \subseteq \clonableQMA'(c, \; c^2 f - 2 \sqrt{(1-c^2)(1-f)}, \; s)$$
\end{proposition}
\begin{proof}
We show both containments as follows:

\paragraph{$\clonableQMA'(c, f, s) \subseteq \clonableQMA(c, f, s)$:}
In this direction, we have a single polynomial time quantum Turing machine, $V_C$, that when given a single copy of a valid witness, $\ket{\psi}$, it both verifies the witness with probability at least $c$ and outputs a quantum state that has fidelity $f$ with two copies of $\ket{\psi}$. When the instance is a $\NO$ instance no purported witness causes it to accept with probability greater than $s$. If we have such a one-shot verifier-cloner, we can use it separately as either the verifier or the cloner that is required of the standard definition of $\clonableQMA$ with the same parameters.

\paragraph{$\clonableQMA(c, f, s) \subseteq \clonableQMA'(c, \; c^2 f - 2 \sqrt{(1-c^2)(1-f)}, \; s)$:}
In this direction, we start out with a separate verifier $V$ and cloner $C$. We have that for every $\YES$ instance, $x$, there is a witness $\ket{\psi}$ that $V$ accepts with probability $c$ and $C$ clones with fidelity $f$. We also have that for every $\NO$ instance, there is no purported witness that would cause $V$ to accept with probability greater than $s$.

The idea is that we first run the verifier in superposition (that is, in a unitary fashion, delaying any measurements) on the witness $\ket{\psi}$, measure and output its output register, and then run the verifier in reverse to uncompute and recover a mixed state $\tilde \rho$ that is close to the original witness. We then run the cloner on the resulting state $\tilde\rho$, and output the result.

The combined process of computing, measuring, and uncomputing the verifier can be written as a binary projective measurement on the witness, with measurement operators $\Pi_v$ and $\mathbb I - \Pi_v$, for accept and reject, respectively.

For a $\YES$ instance, $x$, let $\ket{\psi}$ be the corresponding witness, and let $\rho = \ket\psi\bra\psi$. Since the verifier accepts on witness $\ket\psi$ with probability $c$, we have that $\tr(\Pi_v \rho) = \bra\psi \Pi_v \ket\psi \ge c$. 

The state after performing the measurement is $\tilde\rho = \Pi_v \rho \Pi_v + (\mathbb I - \Pi_v)\rho(\mathbb I - \Pi_v)$. Now the fidelity of $\tilde\rho$ with the original witness is
\begin{align*}
    \bra\psi \tilde\rho \ket\psi
    &= \bra\psi \Pi_v \rho \Pi_v \ket\psi + \bra\psi(\mathbb I - \Pi_v)\rho(\mathbb I - \Pi_v)\ket\psi \\
    &\ge \bra\psi \Pi_v \rho \Pi_v \ket\psi \\
    &= \bra\psi \Pi_v \ket\psi\bra\psi \Pi_v \ket\psi \\
    &= \left(\bra\psi \Pi_v \ket\psi\right)^2 \\
    &\ge c^2
\end{align*}

The process of applying the cloner and measuring if the result is indeed $\ket\psi\otimes\ket\psi$ can likewise be written as a binary projective measurement on the witness, with measurement operator $\Pi_c$ and $\mathbb I - \Pi_c$, for success and failure, respectively.
Since the cloner succeeds on witness $\ket\psi$ with probability $f$,
we have that $\tr(\Pi_c \rho) = \bra\psi \Pi_c \ket\psi \ge f$.

Now, since $\bra\psi \tilde\rho \ket\psi \ge c^2$ and $\bra\psi \Pi_c \ket\psi \ge f$, we have by Lemma~\ref{lemma:composition} that $\tr(\Pi_c \tilde\rho) \ge c^2 f - 2 \sqrt{(1-c^2)(1-f)}$. In other words, the verification accepts with probability at least $c$, and the cloning succeeds with cloning fidelity at least $c^2 f - 2 \sqrt{(1-c^2)(1-f)}$.

On the other hand, if $x$ is a $\NO$ instance, we have that for every purported witness $\ket{\psi^*}$, the original verifier $V$ accepts with probability at most $s$. The probability of verification accepting is not affected by anything that happens after it, so it still accepts with probability at most $s$.

We now have a combined verifier-cloner process that runs in polynomial time, and which achieves completeness $c$, soundness $s$, and cloning fidelity $c^2 f - 2 \sqrt{(1-c^2)(1-f)}$.
\end{proof}

\end{document}